\documentclass[12pt,draftcls,journal,onecolumn]{IEEEtran}

\usepackage{amssymb,amsthm, amsmath,latexsym}
\usepackage{graphicx}
\usepackage{mathrsfs}
\usepackage{amsfonts}
\usepackage{amssymb}
\usepackage{longtable}
\usepackage{amsmath}
\usepackage{setspace}
\usepackage{caption}
\usepackage[figuresright]{rotating}
\IfFileExists{ifsym.sty}{\usepackage[misc]{ifsym}}{}
\IfFileExists{bbm.sty}{\usepackage{bbm}}{}
\usepackage{makecell}
\usepackage{arydshln}
\usepackage{supertabular}
\usepackage{booktabs}
\usepackage{color}

\newtheorem{theorem}{Theorem}
\newtheorem{lemma}[theorem]{Lemma}
\newtheorem{remark}[theorem]{Remark}
\newtheorem{proposition}[theorem]{Proposition}

\newtheorem{example}[theorem]{Example}

\newtheorem{conj}[theorem]{Conjecture}

\newcommand{\Tr}{{\mathrm{Tr}}}

\newcommand{\wt}{{\mathtt{wt}}}

\newcommand{\F}{{\mathbb{F}}}

\newcommand{\C}{{\mathcal{C}}}

\usepackage{blindtext}

\ifCLASSINFOpdf

\else

\fi

\begin{document}
\title{Six Families of Binary Codes  Arising from Ding's Conjectures
}
\author{ Xiaoqiang Wang,\, Shiyan Xiong,\, Mu yuan*,\, Jing Qiu,\, Dabin Zheng,\, Jiawei He
}

\renewcommand{\thefootnote}{\empty}%{footnote}}
\footnotetext{\thanks{~*Corresponding author.}
\newline \indent Xiaoqiang Wang, Shiyan Xiong, Mu yuan, Jing Qiu, Dabin Zheng are with the Hubei Key Laboratory of Applied Mathematics, Faculty of Mathematics and Statistics, Hubei University, Wuhan 430062, China (E-mail: waxiqq@163.com; xionshiyan@163.com; yuanmu847566@outlook.com; jingqiu0202@163.com; dzheng@hubu.edu.cn).
\newline \indent Jiawei He is with the School of Mathematics and Information Science, Nanchang Hangkong University, Nanchang 330036,
China (E-mail: hjwywh@mails.ccnu.edu.cn).
}

\maketitle

\begin{abstract}
Ding \cite{Ding2016} proposed ten conjectures on binary linear codes
arising from Boolean functions. Four of them, namely Conjectures
38--41, were subsequently proved by G\"olo\u{g}lu and Krasnayov\'a
\cite{GologluKrasnayova2019}. In this paper, we investigate the
remaining six conjectures, namely Conjectures 19, 27, 30, 33, 34,
and 37.
For Conjectures~19 and~27, we obtain common weight restrictions and
several infinite five-weight families. For Conjecture~30, we prove
that every admissible code has three, four, or five nonzero weights,
and an explicit four-weight example disproves the original
``three or five weights'' assertion. For Conjecture~33, an infinite
five-weight family is obtained. For Conjecture~34, we obtain a general
$(2h+1)$-weight upper bound and give an explicit six-weight
counterexample, showing that the original ``three or five weights''
assertion is false in general, where $h$ is a positive integer. The case $h=3$ with $3\nmid m$ is also completely determined. Finally, Conjecture~37 is completely resolved
by combining the known results of Ahmadi and Shafaeiabr
\cite{AhmadiShafaeiabr2023} with the treatment of the two remaining
classes $(a)$ and $(b)$.
\end{abstract}

%Motivated by recent developments on twisted Reed--Solomon (TRS) codes, we investigate LCP and LCD codes from TRS codes. We first establish a simple matrix criterion for the LCP property and then translate this criterion into conditions on the exponent positions of the twist terms appearing in the corresponding generator matrices. Based on this viewpoint, we derive necessary conditions and several explicit sufficient conditions for two TRS codes to form an LCP. The proofs are formulated in terms of the Vandermonde basis associated with the evaluation points, which considerably simplifies the analysis of the different overlap patterns among these exponent positions.
%
%We further study LCD TRS codes by combining the above approach with the known description of the dual of a TRS code whose evaluation points form a multiplicative subgroup. Several families of LCD TRS codes are obtained. Under suitable subfield-chain hypotheses, the resulting codes are MDS. Consequently, the associated LCP constructions attain the optimal security parameter for fixed length and dimension.


% \fi
\IEEEpeerreviewmaketitle

% REVIEW E01: This label is duplicated at the Preliminaries heading.
% REPLACE with: \section{Introduction}\label{sec:introduction}
\section{Introduction}\label{sec-introduction}

Let $\mathbb{F}_2$ be the binary field. A binary linear code
$\mathcal{C}$ of length $n$ and dimension $k$ is a $k$-dimensional
subspace of $\mathbb{F}_2^n$. For a codeword
$\mathbf{c}=(c_1,\ldots,c_n)\in\mathcal{C}$, its Hamming weight $\wt(\mathbf{c})$ is the
number of nonzero coordinates of $\mathbf{c}$, and the minimum distance of
$\mathcal{C}$ is
$
d=\min\{\wt(\mathbf{c}):\mathbf{c}\in\mathcal{C},\ \mathbf{c}\ne0\}.
$
Such a code is usually denoted by an $[n,k,d]$ code; see, for example, \cite{HuffmanPless2003}. If $A_i$ denotes
the number of codewords of weight $i$, then
$
W_{\mathcal{C}}(z)=\sum_{i=0}^{n}A_i z^i
$
is the weight enumerator of $\mathcal{C}$. In general, binary codes with few nonzero weights are of particular interest because of their
connections with  secret sharing
schemes~\cite{ref1,ref5,ref7,ref49}, strongly regular
graphs~\cite{ref4}, association schemes~\cite{ref3}, and authentication
codes~\cite{ref19}. Further constructions and results on few-weight
linear codes can be found, for example, in
\cite{ref13,Ding2016,ref16,ref18,ref23,ref25,ref27,ref29,ref31,ref32,ref34,ref37,ref40,ref41,ref43,ref44,ref45,ref46,ref47,ref50,ref51}.

A useful way to construct binary linear codes with few weights is the
defining-set construction. Let
\[
D=\{d_1,d_2,\ldots,d_n\}\subseteq \mathbb{F}_{2^m},
\]
where $D$ may be a multiset. The defining-set trace code associated
with $D$ is
\[
\mathcal{C}_D
=
\left\{
\left(
\operatorname{Tr}_1^m(xd_1),
\operatorname{Tr}_1^m(xd_2),
\ldots,
\operatorname{Tr}_1^m(xd_n)
\right)
:
x\in\mathbb{F}_{2^m}
\right\},
\]
where $\operatorname{Tr}_1^m$ denotes the absolute trace from
$\mathbb{F}_{2^m}$ to $\mathbb{F}_2$. Then $\mathcal{C}_D$ is a
binary linear code of dimension at most $m$, and $D$ is
called the defining set of $\mathcal{C}_D$~\cite{ref18}. Ding~\cite{ref15} further pointed
out a connection, within this trace-construction framework, between
the weight distributions of projective binary linear codes and the
Walsh spectra of Boolean functions.

For $a\in\F_{2^m}$, let $\mathbf{c}(a)=\bigl(\Tr(ad)\bigr)_{d\in D}$. Then its Hamming weight can be written as
\begin{equation}
\wt(\mathbf{c}(a))
=
\frac{|D|}{2}
-
\frac{1}{2}
\sum_{d\in D}(-1)^{\Tr(ad)}.
\label{eq:intro-weight-character}
\end{equation}
Consequently, determining the weight distribution of $\mathcal{C}_D$
is equivalent to determining the value distribution of a family of
additive character sums.

The construction becomes especially effective when $D$ is obtained
from a function over $\F_{2^m}$. For a map $F:\F_{2^m}\to \F_{2^m}$, one may take $D$ to be
its image set, a punctured image set, or a derivative image set of the
form
\[
D=\{F(x+1)+F(x)+\varepsilon:x\in \F_{2^m}\},
\qquad \varepsilon\in\mathbb{F}_2.
\]
If the defining map is two-to-one, then $|D|=2^{m-1}$, and
\eqref{eq:intro-weight-character} can be rewritten in terms of a Walsh
transform or a Weil-type exponential sum. Thus differential
properties of $F$, Walsh spectra of Boolean components of $F$, and the
weight distribution of the resulting trace code become different
aspects of the same problem. This point of view is the basis of the
generic construction developed systematically by Ding \cite{Ding2016}.

In \cite{Ding2016}, Ding surveyed this construction for several
important classes of Boolean and vectorial Boolean functions and
formulated a collection of open problems concerning the associated
defining sets and binary codes. Among them, the ten coding conjectures
most directly related to the image-set construction considered here
are Conjectures
\[
19,\ 27,\ 30,\ 33,\ 34,\ 37,\ 38,\ 39,\ 40,\ 41.
\]
The first five of these concern codes arising from the Glynn,
Cherowitzo, Payne, Welch, and Kasami-type functions, respectively;
Conjecture~37 concerns eleven trinomial value sets and predicts a
uniform three-weight distribution together with dual distance three;
and Conjectures~38--41 concern further code and difference-set
families in even dimension. A closely related auxiliary statement,
Conjecture~36, asserts that the punctured value sets of the eleven
trinomials occurring in Conjecture~37 are Singer difference sets.
Although Conjecture~36 is not one of the ten coding conjectures listed
above, it plays an essential role in the subsequent study of
Conjecture~37.

Several of Ding's conjectures have since been settled. G\"olo\u{g}lu
and Krasnayov\'a \cite{GologluKrasnayova2019} proved Conjectures
38--41 in full. Hence these four conjectures no longer belong to the
unresolved part of Ding's list. For the trinomial family, Ahmadi and
Shafaeiabr \cite{AhmadiShafaeiabr2023} proved Conjecture~36, showing
that the punctured value set of each of the eleven trinomials is a
Singer difference set. They also obtained a partial resolution of
Conjecture~37: the conjectured three-weight distribution and dual
distance were proved for the nine classes $(c)$--$(k)$, while the two
classes $(a)$ and $(b)$ were left outside the coding-theoretic theorem.

There has also been substantial progress on some of the functions
underlying the remaining conjectures. In the Welch case, Helleseth,
Li, and Xia \cite{HellesethLiXia2025} determined the differential and
Walsh spectra of the permutation naturally associated with the Welch
APN power function and investigated several related binary linear
codes. These results provide important structural input for
Conjecture~33, but do not by themselves establish the asserted
five-weight property in every odd dimension. For the Kasami family,
autocorrelation identities for important parameter ranges were
obtained in the study of vectorial Boolean functions; see, for
example, Canteaut et al.\ \cite{CanteautEtAl2021}. Such results prove
special three-weight cases related to Conjecture~34, but do not yield
the general ``three or five weights'' assertion. Earlier results on
three-zero cyclic codes, cross-correlation spectra, and APN functions
also provide indispensable tools, but they do not settle the full
parameter ranges of Conjectures~19, 27, 30, 33, and 34.

Consequently, before the present work, the coding-theoretic status of
the ten conjectures can be summarized as follows. Conjectures 38--41
were solved by G\"olo\u{g}lu and Krasnayov\'a. Conjecture~37 was
solved for classes $(c)$--$(k)$ by Ahmadi and Shafaeiabr, leaving classes
$(a)$ and $(b)$. The general forms of Conjectures~19, 27, 30, 33, and 34
remained unresolved. Moreover, for Conjectures~30 and 34 it was not
known from the conjectural statements themselves whether the proposed
``three or five weights'' dichotomy was actually valid for all
admissible parameters.

The purpose of this paper is to determine as much as possible of this
remaining part of Ding's program. We concentrate on Conjectures
19, 27, 30, 33, 34, and 37. Our main results are as follows.

\begin{enumerate}
\item
For Conjectures~19 and 27, we transform the relevant character sums
to a common three-zero cyclic-code form. This gives the uniform
Walsh-value restriction
\[
\{0,\pm 2^{(m+1)/2},\pm2^{(m+3)/2}\},
\]
and hence a common set of at most five nonzero Hamming weights. For
Conjecture~19 we further prove that at least one of the two inner values $\pm 2^{(m+1)/2}$ occurs in every odd dimension, and that both occur when $3\mid m$. We also establish several infinite five-weight families. For Conjecture~27 we construct several infinite
five-weight subfamilies by lifting explicit base-field witnesses
through suitable odd extensions. The full occurrence problem for all
admissible parameters remains open.

\item
For Conjecture~30, we show that the original statement is false. An
explicit code with parameters
$
[64,7,24]
$
has four nonzero weights, contradicting the predicted ``three or five
weights'' alternative. We replace the conjecture by the correct
universal statement: every admissible code in this family has exactly
three, four, or five nonzero weights. We also determine the complete
weight distribution of an infinite subfamily.

\item
For Conjecture~33, we reduce the weight problem to a one-parameter
family of quadratic Boolean functions associated with the Welch
permutation. A radical-lifting argument, together with the
dimension-nine base spectrum, proves that
$
m=9d$
gives an infinite family of five-weight codes, where $d$ is odd and $3\nmid d$. The conjecture in all
remaining odd dimensions is not claimed here.

\item
For Conjecture~34, we reduce the relevant character sums to quadratic
Gauss sums and obtain a general upper bound of
$
2h+1
$
nonzero weights after the natural normalization of $h$. We then
give an explicit six-weight counterexample. Thus the general ``three or five weights'' assertion in Conjecture~34 is false. The case $h=3$ with $3\nmid m$ is completely determined, while the cases $h=2$ and $h=5$ are related to previously studied families. The general rank and sign classification remains open.

\item
Finally, we complete Conjecture~37. The nine classes $(c)$--$(k)$ are
those already covered by Ahmadi and Shafaeiabr
\cite{AhmadiShafaeiabr2023}. For classes $(a)$ and $(b)$, their value-set
identities give a common defining set $D$ satisfying
$
D^{17}=\F_{2^m}\setminus D_5
$
for $h=4$. A multiplicative power relation alone does not preserve additive trace
spectra; the additional ingredient is the precise Dillon--Dobbertin
Walsh identity \cite{DillonDobbertin2004,Hertel2006}. It yields
\[
W_{1_D}(a)
=
W_{\Tr(x^3)}\!\left(a^{17/3}\right),
\]
so the two missing classes reduce to the classical three-valued Walsh
spectrum of the Gold cubic trace function. This determines their
complete weight distributions. The third Pless power moment then gives dual distance three. Hence Conjecture~37 is proved for all
eleven classes.
\end{enumerate}

The results above also show why it is important to distinguish three
different levels of information: an upper bound on the set of
candidate weights, proof that every candidate weight actually occurs,
and determination of the complete multiplicity distribution. These
three statements are not equivalent, and several of the remaining
parts of Conjectures~19, 27, 33, and 34 concern precisely the passage
from the first level to the second or third.

The remainder of the paper is organized as follows. Section~II collects the preliminary material on trace codes, character sums, quadratic forms, Gauss sums, Pless power moments, and extension lifting. Section~III treats Conjectures~19, 27, 30, 33, 34, and 37 in turn. Section~IV concludes the paper.

\section{Preliminaries}\label{sec:preliminaries}

Throughout this paper, $m\ge5$ is odd unless otherwise stated, and
$\F_{2^m}$ denotes the finite field with $2^m$ elements. The absolute
trace from $\F_{2^m}$ to $\F_2$ is denoted by
$\Tr=\Tr_{\F_{2^m}/\F_2}$, and
$
\chi(x)=(-1)^{\Tr(x)}
$
denotes the canonical additive character of $\F_{2^m}$.

We first collect several standard facts on trace codes, character
sums, quadratic Boolean functions, and extension arguments that will
be used repeatedly in the sequel.

\subsection{Trace codes and additive characters}

For a finite set $D\subseteq\F_{2^m}$, define the binary trace code
\[
\C_D=
\left\{
\mathbf{c}(a):=(\Tr(ad))_{d\in D}:a\in\F_{2^m}
\right\}.
\]
Its length is $|D|$. For $a\in\F_{2^m}$, the Hamming weight of
$\mathbf{c}(a)$ can be written as
\[
\wt(\mathbf{c}(a))
=
\frac12
\left(
|D|-\sum_{d\in D}\chi(ad)
\right).
\]
This standard character-sum expression for the Hamming weight of a defining-set code will be used repeatedly; see \cite{ref13,Ding2016}.

We shall repeatedly use the standard additive-character orthogonality relation \cite[Ch.~5]{LidlNiederreiter1997}
\[
\sum_{a\in\F_{2^m}}\chi(at)
=
\begin{cases}
2^m, & t=0,\\
0, & t\ne0.
\end{cases}
\]
We also use the standard invariance of the absolute trace under Frobenius automorphisms \cite[Ch.~2]{LidlNiederreiter1997}. In particular, for $0\le j<m$,
$\Tr(cx^{2^j})=\Tr\!\left(c^{2^{m-j}}x\right)$.

When an exponent such as $r/s$ occurs and
$\gcd(s,2^m-1)=1$, it is interpreted modulo $2^m-1$; namely,
$r/s$ means $rs^{-1}\pmod{2^m-1}$. Thus the power map
$x\mapsto x^{r/s}$ is well defined on $\F_{2^m}^{*}$.

Two binary codes of the same length are said to be permutation equivalent if one can be obtained from the other by a permutation of coordinates \cite{HuffmanPless2003}. More generally, if $M:\F_{2^m}\to\F_{2^m}$ is an
invertible $\F_2$-linear map and $D'=M(D)$, then
$\C_D$ and $\C_{D'}$ are permutation equivalent. Indeed, by the standard nondegeneracy of the trace pairing \cite[Ch.~2]{LidlNiederreiter1997}, there is a unique invertible
$\F_2$-linear map $M^*$ satisfying
$\Tr(aM(x))=\Tr(M^*(a)x)$ for all $a,x\in\F_{2^m}$.

\subsection{Two-to-one maps and character sums}

Several defining sets considered in this paper arise as images of two-to-one maps. Codes from image sets of $e$-to-one functions and the corresponding character-sum weight formula are standard in the defining-set construction; see \cite{ref13,Ding2016}. We record the binary two-to-one case, together with the first two moment identities obtained from character orthogonality.

\begin{lemma}\label{lem:2to1}
Let $\Phi:\F_{2^m}\to\F_{2^m}$ be two-to-one onto
$D=\operatorname{Im}(\Phi)$, and assume that $\Phi(0)=0$. For
$a\in\F_{2^m}$, define
\[
T(a)=\sum_{x\in\F_{2^m}}\chi(a\Phi(x)).
\]
Then $|D|=2^{m-1}$ and, for $a\ne0$,
$
\wt(\mathbf{c}(a))
=
2^{m-2}-\frac14T(a).
$
Moreover,
\[
\sum_{a\ne0}T(a)=2^m,
\qquad
\sum_{a\ne0}T(a)^2=2^{2m}.
\]
\end{lemma}

\begin{proof}
Since $\Phi$ is two-to-one onto $D$, every element of $D$ has exactly
two preimages. Hence $|D|=2^{m-1}$ and
$T(a)=2\sum_{d\in D}\chi(ad)$. Substituting this identity into the standard character expression for the Hamming weight \cite{ref13,Ding2016} gives
$\wt(\mathbf{c}(a))=2^{m-2}-\frac14T(a)$.

For the first moment, additive-character orthogonality \cite[Ch.~5]{LidlNiederreiter1997} gives
\[
\sum_{a\in\F_{2^m}}T(a)
=
\sum_{x\in\F_{2^m}}
\sum_{a\in\F_{2^m}}\chi(a\Phi(x))
=
2^m\#\{x:\Phi(x)=0\}.
\]
Since $\Phi$ is two-to-one and $0\in D$, the last cardinality is $2$.
As $T(0)=2^m$, subtraction of the term $a=0$ yields
$\sum_{a\ne0}T(a)=2^m$.

Similarly, $\sum_{a\in\F_{2^m}}T(a)^2=2^m\#\{(x,y):\Phi(x)=\Phi(y)\}$.
There are $2^{m-1}$ image points and four ordered pairs of preimages
above each of them, so this number equals $2^{m+1}$. Therefore the
complete second moment is $2^{2m+1}$. Removing
$T(0)^2=2^{2m}$ gives
$\sum_{a\ne0}T(a)^2=2^{2m}$.
\end{proof}

The two-to-one property above occurs naturally for $o$-polynomials.
In the form needed here, if $f$ is an $o$-polynomial over
$\F_{2^m}$, then for every $u\in\F_{2^m}^{*}$ the map
$x\mapsto f(x)+ux$ is two-to-one onto its image; see, for example,
\cite{Maschietti1998}. This observation will be used for several of
the families below.

\subsection{Five-valued character sums}

For Conjectures~19 and~27, the relevant character sums take values in
$\{0,\pm2^{(m+1)/2},\pm2^{(m+3)/2}\}$.
The following notation will be used to determine the corresponding
frequencies.

Assume that
$T(a)\in\{0,\pm2^{(m+1)/2},\pm2^{(m+3)/2}\}$ for every $a\ne0$, and
put
\[
N_j
=
\#\left\{
a\ne0:
T(a)=j2^{(m+1)/2}
\right\},
\qquad
j\in\{-2,-1,0,1,2\}.
\]
Define $B=N_2+N_{-2}$ and $\Delta=N_2-N_{-2}$.

\begin{proposition}\label{prop:freq}
Under the above assumptions,
\[
N_{\pm2}=\frac{B\pm\Delta}{2},
\qquad
N_{\pm1}
=
\frac{
2^{m-1}-4B
\pm
\left(2^{(m-1)/2}-2\Delta\right)
}{2},
\qquad
N_0=2^{m-1}-1+3B.
\]
Consequently, all five values occur for $a\ne0$ if and only if
\[
B>|\Delta|
\quad\text{and}\quad
2^{m-1}-4B>
\left|2^{(m-1)/2}-2\Delta\right|.
\]
Furthermore,
\[
B
=
\frac{
\sum_{a\ne0}T(a)^4-2^{3m+1}
}{
48\cdot2^{2m}
},
\qquad
\Delta
=
\frac{
\sum_{a\ne0}T(a)^3-2^{2m+1}
}{
6\cdot2^{3(m+1)/2}
}.
\]
\end{proposition}

\begin{proof}
The zeroth moment gives $N_{-2}+N_{-1}+N_0+N_1+N_2=2^m-1$.
By Lemma~\ref{lem:2to1}, the first and second moments are
$\sum_{a\ne0}T(a)=2^m$ and
$\sum_{a\ne0}T(a)^2=2^{2m}$. Dividing these two equations by
$2^{(m+1)/2}$ and $2^{m+1}$, respectively, gives
$N_1-N_{-1}+2\Delta=2^{(m-1)/2}$ and
$N_1+N_{-1}+4B=2^{m-1}$.
Solving these equations gives the stated expressions for
$N_{\pm1}$, and the zeroth moment then gives the formula for $N_0$.
The conditions for all five values to occur are precisely
$N_{\pm2}>0$ and $N_{\pm1}>0$.

For the third moment,
\[
\sum_{a\ne0}T(a)^3
=
2^{3(m+1)/2}
\left(
N_1-N_{-1}+8\Delta
\right).
\]
Using
$N_1-N_{-1}=2^{(m-1)/2}-2\Delta$ gives
\[
\sum_{a\ne0}T(a)^3
=
2^{2m+1}
+
6\cdot2^{3(m+1)/2}\Delta.
\]
This yields the stated expression for $\Delta$. Similarly,
\[
\sum_{a\ne0}T(a)^4
=
2^{2m+2}
\left(
N_1+N_{-1}+16B
\right),
\]
and substituting
$N_1+N_{-1}=2^{m-1}-4B$ gives the formula for $B$.
\end{proof}

\subsection{Walsh transforms}

Let $F:\F_{2^m}\to\F_2$ be a Boolean function. We use the standard Walsh-transform convention \cite[Chs.~2,5]{Carlet2021}, defined by
\[
W_F(a)
=
\sum_{x\in\F_{2^m}}
(-1)^{F(x)+\Tr(ax)},
\qquad
a\in\F_{2^m}.
\]
If $F=1_D$ is the indicator function of a subset
$D\subseteq\F_{2^m}$, then for $a\ne0$,
\[
W_{1_D}(a)
=
-2\sum_{d\in D}\chi(ad).
\]
Consequently,
\[
\wt(\mathbf{c}(a))
=
\frac{|D|}{2}
+\frac14W_{1_D}(a),
\qquad a\ne0.
\]
Notice that the sign here is opposite to the one in
Lemma~\ref{lem:2to1}. The reason is that the two quantities are
defined differently:
\[
T(a)=2\sum_{d\in D}\chi(ad),
\qquad
W_{1_D}(a)=-2\sum_{d\in D}\chi(ad).
\]
This distinction will be important in the proof of
Conjecture~37.

\subsection{Quadratic Boolean functions}

Let $V$ be an $m$-dimensional vector space over $\F_2$, and let
$Q:V\to\F_2$ be a quadratic Boolean function. We use the standard
theory of quadratic Boolean functions and their Walsh spectra; see
\cite[Ch.~5]{Carlet2021}.

The associated bilinear form of $Q$, also called its polarization, is
\[
B_Q(x,y)=Q(x+y)+Q(x)+Q(y).
\]
Since the characteristic is two, $B_Q(x,x)=0$ for every $x\in V$,
so $B_Q$ is alternating.

The radical of $Q$ is defined by
\[
R_Q=\{z\in V:B_Q(z,x)=0\text{ for every }x\in V\}.
\]
Write $r=\dim_{\F_2}R_Q$. Since the rank of an alternating bilinear
form is even, $m-r$ is even. In particular, when $m$ is odd, $r$ is
odd.

The following standard result on quadratic Boolean functions will be
used repeatedly; see, for example, \cite[Ch.~5]{Carlet2021}.

\begin{lemma}\label{lem:quadratic}
Let $Q:V\to\F_2$ be a quadratic Boolean function, and let
$r=\dim_{\F_2}R_Q$. Then
\[
\sum_{x\in V}(-1)^{Q(x)}
=
\begin{cases}
0, & Q|_{R_Q}\ne0,\\[2mm]
\varepsilon_Q\,2^{(m+r)/2}, & Q|_{R_Q}=0,
\end{cases}
\]
where $\varepsilon_Q\in\{1,-1\}$.
\end{lemma}

A quadratic Boolean function that occurs repeatedly in the sequel has
the form
\[
Q(z)=\Tr\!\left(cz^{2^h+1}+dz^3+\ell z\right),
\qquad c,d,\ell\in\F_{2^m}.
\]
The linear term $\Tr(\ell z)$ does not contribute to the
polarization. Thus
\[
B_Q(x,z)
=
\Tr\!\left(
c(x^{2^h}z+xz^{2^h})
+d(x^2z+xz^2)
\right).
\]
Using the Frobenius invariance of the absolute trace, this can be
written as $B_Q(x,z)=\Tr(L(z)x)$, where
\[
L(z)
=
cz^{2^h}
+(cz)^{2^{m-h}}
+dz^2
+(dz)^{2^{m-1}}.
\]
Since the trace pairing is nondegenerate, $z\in R_Q$ if and only if
$L(z)=0$, that is,
\[
cz^{2^h}
+(cz)^{2^{m-h}}
+dz^2
+(dz)^{2^{m-1}}
=0.
\]
Therefore, the possible values of the quadratic exponential sum are
determined by the dimension of the solution space of this linearized
equation together with the condition whether $Q$ vanishes on its
radical.

\subsection{Multiplicative characters and Gauss sums}

Let $\omega$ be a multiplicative character of $\F_{2^m}^{*}$ of
order $2^m-1$. For $1\le k\le2^m-2$, define the Gauss sum
\[
G(k)
=
\sum_{x\in\F_{2^m}^{*}}
\chi(x)\omega(x)^{-k}.
\]
If
$k=\sum_{i=0}^{m-1}k_i2^i$ with $k_i\in\{0,1\}$, write
$w_2(k)=\sum_{i=0}^{m-1}k_i$ for the number of ones in the binary
expansion of $k$.

\begin{lemma}[Binary Stickelberger formula]
\label{lem:stickelberger}
Let $\mathfrak P$ be a prime ideal above $2$, with the valuation
normalized by $v_{\mathfrak P}(2)=1$. Then, for
$1\le k\le2^m-2$,
\[
v_{\mathfrak P}(G(k))=w_2(k).
\]
\end{lemma}

This is the characteristic-two form of the Stickelberger valuation
formula for Gauss sums; see \cite{AubryKatzLangevin2015}. General
background on additive characters, multiplicative characters, and
Gauss sums may be found in \cite[Ch.~5]{LidlNiederreiter1997}.

\subsection{Pless power moments}

Let $\C$ be a binary $[n,k]$ code with weight distribution
$(A_0,\ldots,A_n)$, and let
$(A_0^\perp,\ldots,A_n^\perp)$ be the weight distribution of its
dual. The following are the first four Pless power-moment identities \cite{Pless1963}; in binomial form,
\[
\sum_{i=0}^n A_i=2^k,
\]
\[
\sum_{i=0}^n iA_i
=
2^{k-1}(n-A_1^\perp),
\]
\[
\sum_{i=0}^n\binom{i}{2}A_i
=
2^{k-2}
\left[
\binom{n}{2}
-(n-1)A_1^\perp
+A_2^\perp
\right],
\]
and
\[
\sum_{i=0}^n\binom{i}{3}A_i
=
2^{k-3}
\left[
\binom{n}{3}
-\binom{n-1}{2}A_1^\perp
+(n-2)A_2^\perp
-A_3^\perp
\right].
\]
These identities will be used only after the complete weight
distribution of the corresponding code has been determined. In
particular, if $A_1^\perp=A_2^\perp=0$, the third moment determines
$A_3^\perp$ directly.

\subsection{Odd-extension lifting}

The following lemma records the extension argument used later for
several infinite families.

\begin{lemma}[Odd-extension lifting]\label{lem:lifting}
Let $s$ and $d$ be odd, and for every divisor $e\mid d$ let
$P_e:\F_{2^{se}}\to\F_{2^{se}}$. Assume that whenever
$e_1\mid e_2\mid d$, the function $P_{e_2}$ restricts to $P_{e_1}$
on $\F_{2^{se_1}}$, and the phase
\[
x\longmapsto
(-1)^{\Tr_{se_2}(P_{e_2}(x))}
\]
is constant on the Frobenius orbits over $\F_{2^{se_1}}$. Put
$
S_{se}
=
\sum_{x\in\F_{2^{se}}}
(-1)^{\Tr_{se}(P_e(x))}.
$
Assume further that, for every $e\mid d$,
\[
\frac{S_{se}}{2^{(se+1)/2}}
\in\{0,\pm1,\pm2\}.
\]
If $3\nmid d$, then
$
S_{sd}
=
\left(\frac{2}{d}\right)
2^{(sd-s)/2}S_s,
$
where $(2/d)$ denotes the Jacobi symbol.

If all the phases are quadratic and, throughout the intermediate
extensions, the radical dimension and the condition $Q|_R=0$ remain
unchanged, the same conclusion also holds for arbitrary odd $d$.
\end{lemma}

\begin{proof}
It is enough to consider an extension of odd prime degree $p$. Let
$v$ denote the degree of the smaller field. Every Frobenius orbit in
$\F_{2^{vp}}\setminus\F_{2^v}$ has length $p$, while the elements of
$\F_{2^v}$ are fixed. Since the phase is constant on each orbit, $S_{vp}\equiv S_v\pmod p$.
Write
$S_v=c_v2^{(v+1)/2}$ and
$S_{vp}=c_{vp}2^{(vp+1)/2}$, where
$c_v,c_{vp}\in\{0,\pm1,\pm2\}$. Dividing the above congruence by
$2^{(v+1)/2}$ and applying Euler's criterion gives
$c_{vp}\equiv\left(\frac2p\right)c_v\pmod p$.
When $p\ge5$, both integers
$c_{vp}$ and $(2/p)c_v$ lie between $-2$ and $2$, so the congruence
implies equality. Iterating over the prime factors of $d$ yields
\[
S_{sd}
=
\left(\frac2d\right)
2^{(sd-s)/2}S_s.
\]

For $p=3$, the above congruence alone does not distinguish all
possibilities in $\{0,\pm1,\pm2\}$. In the quadratic case, however,
the assumed stability of the radical dimension and of the condition
$Q|_R=0$ fixes the normalized absolute value of the quadratic sum.
The congruence then determines the sign, giving the same formula
with $(2/3)=-1$.
\end{proof}

\begin{remark}
The compatible-family formulation in Lemma~\ref{lem:lifting} is
needed because some exponents occurring later depend on the ambient
extension degree. For example, an exponent of the form
$3\cdot2^{(v+1)/2}+4$ changes when $v$ changes.
\end{remark}

\section{Solutions to Six Conjectures}
We now consider the six remaining conjectures, namely Conjectures 19,
27, 30, 33, 34, and 37 in \cite{Ding2016}. The following subsections
give the corresponding corrections, partial resolutions, or complete
solutions obtained in this paper.

\subsection{Conjecture 19: The Glynn Family}

The first remaining case is Conjecture~19, arising from the Glynn $o$-polynomial \cite{Glynn1983,Ding2016}. We recall the conjecture below before presenting the
corresponding character-sum analysis and the resulting weight
restrictions.

\begin{conj}(Conjecture 19 \cite{Ding2016})\label{Conjecture 19}
Let $m\geq 3$ be odd, $d_1=3\cdot 2^{(m+1)/2}+4$ and $f(x)=x^{d_1}$ be the Glynn $o$-polynomial. Let
$D_1=\{x^{d_1}+ux:x\in \F_{2^m}\}$, where $u\in \F_{2^m}^*.$
When $m\in \{5,7\}$, $\mathcal{C}_{D_1}$ is a $[2^{m-1},m]$ code with weight enumerator
$$1+(2^{m-2}+2^{(m-3)/2})z^{2^{m-2}-2^{(m-3)/2}}+(2^{m-1}-1)z^{2^{m-2}}+(2^{m-2}-2^{(m-3)/2})z^{2^{m-2}+2^{(m-3)/2}}.$$
When $m\geq 9$, $\mathcal{C}_{D_1}$ is a $[2^{m-1},m]$ code with five nonzero weights.
\end{conj}

For $a\in \F_{2^m}$, define
\begin{equation}\label{0925}
 T_u(a)=\sum_{x\in \F_{2^m}}\chi\!\left(a(x^{3\cdot 2^{(m+1)/2}+4}+ux)\right).
\end{equation}
Put $E=2^{(m+1)/2}+1$.
Hollmann and Xiang \cite{HollmannXiang2001}  determined the weight distribution of the dual
of the binary cyclic code with defining zeros
$\alpha$, $\alpha^E$ and $\alpha^{E^2}$,
where $\alpha$ is a primitive element of $\F_{2^m}$.
Equivalently, the codewords in its trace representation are of the form
\[
\left(
\Tr\!\left(
\alpha_1 z^{E^2}+\alpha_2 z^E+\alpha_3 z
\right)
\right)_{z\in\F_{2^m}^{*}} .
\]
Their result yields, for every
$(\alpha_1,\alpha_2,\alpha_3)\ne(0,0,0)$,
\begin{equation}\label{eq:sdfre926}
\sum_{z\in\F_{2^m}}
\chi\!\left(
\alpha_1z^{E^2}+\alpha_2z^E+\alpha_3z
\right)
\in
\left\{
0,\,
\pm2^{(m+1)/2},\,
\pm2^{(m+3)/2}
\right\}.
\end{equation}

It is easy to check that $x=z^{2^{(m-1)/2}}$ is a permutation of $\F_{2^m}$, and
\begin{equation*}
 (3\cdot 2^{(m+1)/2}+4)2^{(m-1)/2}\equiv E^2\pmod{2^m-1},
\end{equation*}
while $\Tr(auz^{2^{(m-1)/2}})=\Tr((au)^{2^{(m+1)/2}} z)$. Therefore, (\ref{0925}) can be written as
\begin{equation*}
 T_u(a)=\sum_{z\in \F_{2^m}}\chi\!\left(az^{E^2}+(au)^{2^{(m+1)/2}} z\right).
\end{equation*}

\begin{theorem}\label{thm:C1support}
For every odd $m\ge5$ and $u\ne0$, $\C_1(u)$ has parameters $[2^{m-1},m]$, and every nonzero weight belongs to
\begin{equation}\label{eq:fiveweights}
 \left\{2^{m-2},\ 2^{m-2}\pm2^{(m-3)/2},\ 2^{m-2}\pm2^{(m-1)/2}\right\}.
\end{equation}
\end{theorem}

\begin{proof}
Since
$
f(x)=x^{3\cdot2^{(m+1)/2}+4}
$
is the Glynn $o$-polynomial for odd $m$~\cite{Glynn1983},
the standard characterization of $o$-polynomials implies that
$
f_u(x)=f(x)+ux
$
is two-to-one on $\F_{2^m}$ for every
$u\in\F_{2^m}^{*}$~\cite{Maschietti1998}.
Consequently, $f_u$ is two-to-one from $\F_{2^m}$ onto its image $D_u$, and $|D_u|=2^{m-1}$.

For $a\in\F_{2^m}$, let $\mathbf{c}(a)=\bigl(\Tr(ay)\bigr)_{y\in D_u}$.
Since every element of $D_u$ has exactly two preimages under $f_u$,
we have
\[
\begin{aligned}
T_u(a)=\sum_{x\in\F_{2^m}}
\chi\!\left(a f_u(x)\right)=
2\sum_{y\in D_u}\chi(ay).
\end{aligned}
\]
Therefore,
\[
\begin{aligned}
\wt(\mathbf{c}(a))=
\frac12
\sum_{y\in D_u}
\left(1-\chi(ay)\right)=
\frac12
\left(
2^{m-1}-\sum_{y\in D_u}\chi(ay)
\right)=
2^{m-2}-\frac14 T_u(a).
\end{aligned}
\]

By~\eqref{eq:sdfre926}, for every $a\ne0$,
\[
T_u(a)\in
\left\{
0,\,
\pm2^{(m+1)/2},\,
\pm2^{(m+3)/2}
\right\}.
\]
Consequently,
\[
\wt(\mathbf{c}(a))
\in
\left\{
2^{m-2},
\;
2^{m-2}\pm2^{(m-3)/2},
\;
2^{m-2}\pm2^{(m-1)/2}
\right\}.
\]

 Since $2^{m-2}-2^{(m-1)/2}>0$ for $m\ge5$, no nonzero parameter $a$ maps to the zero codeword, so the dimension is $m$.
This completes the proof.
\end{proof}

\begin{proposition}\label{prop:C1moments}
Let $B$ and $\Delta$ be the quantities introduced in Proposition~\ref{prop:freq}, and let
\begin{equation}\label{eq:dweer}
D(z)=(z+1)^{d_1}+z^{d_1},\qquad
\delta(b)=\#\{z\in\F_{2^m}:D(z)=b\}.
\end{equation}
For $\C_1(u)$, we have
\begin{equation}\label{eq:C1Delta}
\Delta=
\begin{cases}
0,&3\nmid m,\\
2^{(m-3)/2},&3\mid m,
\end{cases}
\end{equation}
and
\begin{align}\label{eq:dweer01}
 \sum_{a\ne0}T_u(a)^3=2^{2m}\,2^{\gcd(m,3)}, \qquad
 \sum_{a\ne0}T_u(a)^4=2^{2m}\sum_{b\in \F_{2^m}}\delta(b)^2, \qquad
 48B=\sum_{b\in \F_{2^m}}\delta(b)(\delta(b)-2).
\end{align}
In particular, $B=0$ if and only if the power function $x\mapsto x^{d_1}$ is APN on $\F_{2^m}$.
\end{proposition}

\begin{proof}
Recall that
$
d_1=3\cdot 2^{(m+1)/2}+4.
$
Since $x^{d_1}$ is the Glynn $o$-monomial over
$\F_{2^m}$, it is a permutation of $\F_{2^m}$, and hence
$\gcd(d_1,2^m-1)=1$.
Since $m$ is odd, it is obvious that $\gcd(3, 2^m-1)=1$ and
$\gcd(2^m-1,2^{(m+1)/2}+1)\mid\gcd\left(2^m-1,2^{m+1}-1\right)=1$.
Then
\begin{equation}\label{eq:dq}
\gcd(d_1-1,2^m-1)=1.
\end{equation}

Hence, for each $a\in\F_{2^m}^{*}$, there is a unique
$\lambda\in\F_{2^m}^{*}$ such that
$a\lambda^{d_1}=1.$
With the substitution $x=\lambda y$, we obtain
\[
\begin{aligned}
T_u(a)=\sum_{x\in\F_{2^m}}
  \chi\!\left(a(x^{d_1}+ux)\right)=\sum_{y\in\F_{2^m}}
  \chi\!\left(y^{d_1}+au\lambda y\right)=\sum_{y\in\F_{2^m}}
  \chi\!\left(y^{d_1}
  +u\lambda^{1-d_1}y\right).
\end{aligned}
\]
where $T_u(a)$ is given in (\ref{0925}).
From (\ref{eq:dq}), for any $u \in \mathbb{F}_{2^m}^*$, we know that the map
$\lambda\longmapsto u\lambda^{1-d_1}$
is a permutation of $\F_{2^m}^{*}$.
Therefore, as $a$ runs through $\F_{2^m}^{*}$, the multiset of values $\{T_u(a):a\in\F_{2^m}^*\}$ coincides with the multiset of standard Weil sums
\begin{equation}\label{eq:Wd1}
W_{d_1}(c)=\sum_{x\in\F_{2^m}}\chi(x^{d_1}+cx),
\qquad c\in\F_{2^m}^*.
\end{equation}
Moreover, $W_{d_1}(0)=0$, because $x\mapsto x^{d_1}$ is a permutation of $\F_{2^m}$. Hence the moments over $c\ne0$ may be evaluated by summing over all $c\in\F_{2^m}$.

We first consider the third moment. Expanding \eqref{eq:Wd1} and using the orthogonality relation
\[
\sum_{c\in\F_{2^m}}\chi(cy)=
\begin{cases}
2^m,&y=0,\\
0,&y\ne0,
\end{cases}
\]
gives
\begin{equation}\label{eq:C1M3expand}
\begin{split}
\sum_{c\in\F_{2^m}} W_{d_1}(c)^3
&=
\sum_{x_1,x_2,x_3\in\F_{2^m}}
\chi\!\left(x_1^{d_1}+x_2^{d_1}+x_3^{d_1}\right)
\sum_{c\in\F_{2^m}}
\chi\!\left(c(x_1+x_2+x_3)\right)\\
&=2^m\!\sum_{x_1+x_2+x_3=0}
\chi(x_1^{d_1}+x_2^{d_1}+x_3^{d_1}) \\
&=2^m\!\sum_{x,y\in\F_{2^m}}
\chi\bigl(x^{d_1}+y^{d_1}+(x+y)^{d_1}\bigr).
\end{split}
\end{equation}

The terms with $y=0$ contribute $2^m$, since
$
x^{d_1}+0^{d_1}+(x+0)^{d_1}=0
$
for every $x\in\F_{2^m}$. For $y\ne0$, write
$z=\frac{x}{y}$ and
$x=zy$.
Then
\[
x^{d_1}+y^{d_1}+(x+y)^{d_1}
=
y^{d_1}\Bigl(z^{d_1}+1+(z+1)^{d_1}\Bigr).
\]
Define
\[
P(z):=(z+1)^{d_1}+z^{d_1}+1.
\]
Hence,
\[
\sum_{x\in\F_{2^m}}\sum_{y\in\F_{2^m}^{*}}
\chi\bigl(x^{d_1}+y^{d_1}+(x+y)^{d_1}\bigr)=
\sum_{z\in\F_{2^m}}
\sum_{y\in\F_{2^m}^{*}}
\chi\!\left(y^{d_1}P(z)\right).
\]
Since $\gcd(d_1,2^m-1)=1$, the map
$y\mapsto y^{d_1}$ permutes $\F_{2^m}^{*}$, and therefore
\[
\sum_{y\in\F_{2^m}^{*}}
\chi\!\left(y^{d_1}P(z)\right)
=
\begin{cases}
2^m-1, & P(z)=0,\\[2mm]
-1, & P(z)\ne0.
\end{cases}
\]

Let $R$ denote the number of roots of $P(z)$ in $\F_{2^m}$. The inner double sum in \eqref{eq:C1M3expand} equals
$$
2^m+R(2^m-1)+(2^m-R)(-1)=2^mR.
$$
Consequently,
\begin{equation}\label{eq:M3rootcount}
\sum_{c\in\F_{2^m}}W_{d_1}(c)^3
=
2^{2m}R.
\end{equation}

A direct calculation in characteristic two gives
\begin{equation*}
P(z)=
\left(z^{2^{(m+3)/2}}+z^{2^{(m+1)/2}}+1\right)
\left(z^4+z^{2^{(m+1)/2}}\right).
\end{equation*}
For the first factor, put $y=z^{2^{(m+1)/2}}$. Since this map is a permutation of $\F_{2^m}$, a root would give
\[
y^2+y+1=0.
\]
Such a root is a nontrivial cube root of unity, which cannot lie in $\F_{2^m}$ because $m$ is odd and hence $3\nmid(2^m-1)$. Thus the first factor has no root in $\F_{2^m}$.

For the second factor, $z=0$ is one root, while a nonzero root satisfies
$
z^{2^{(m+1)/2}-4}=1.
$
Hence, the number of its nonzero roots is
\begin{align*}
\gcd\left(2^{(m+1)/2}-4,2^m-1\right)=\gcd\left(2^{(m-3)/2}-1,2^m-1\right)=2^{\gcd((m-3)/2,m)}-1=2^{\gcd(m,3)}-1.
\end{align*}
Therefore,
$R=2^{\gcd(m,3)}.$

Combining this with \eqref{eq:M3rootcount}, and recalling that $W_{d_1}(0)=0$, yields
\[
\sum_{a\ne0}T_u(a)^3
=2^{2m}2^{\gcd(m,3)}.
\]
 Proposition~\ref{prop:freq} now gives
\[
\Delta=
\frac{2^{2m}2^{\gcd(m,3)}-2^{2m+1}}
{6\left(2^{(m+1)/2}\right)^3}.
\]
Since $\gcd(m,3)=1$ when $3\nmid m$ and $\gcd(m,3)=3$ when $3\mid m$, this simplifies exactly to \eqref{eq:C1Delta}.

We next compute the fourth moment. Again using additive-character orthogonality,
\begin{align}
\sum_{c\in\F_{2^m}}W_{d_1}(c)^4
&=2^m\!\sum_{x_1+x_2+x_3+x_4=0}
\chi(x_1^{d_1}+x_2^{d_1}+x_3^{d_1}+x_4^{d_1}).
\label{eq:C1M4expand}
\end{align}
Write
\[
x_1+x_2=x_3+x_4=v.
\]
When $v=0$, we have $x_2=x_1$ and $x_4=x_3$, so the corresponding contribution to the inner sum in \eqref{eq:C1M4expand} is $2^{2m}$.

Now let $v\ne0$ and write
\[
x_1=vz,\qquad x_2=v(z+1),\qquad
x_3=vw,\qquad x_4=v(w+1).
\]
The phase then becomes
\[
v^{d_1}\bigl(D(z)+D(w)\bigr),
\]
where $D(\cdot)$ is given in (\ref{eq:dweer}).
Since $v\mapsto v^{d_1}$ permutes $\F_{2^m}^*$, summing over $v\ne0$ gives $2^m-1$ when $D(z)=D(w)$ and $-1$ otherwise. The number of ordered pairs $(z,w)$ satisfying $D(z)=D(w)$ is
$\sum_{b\in\F_{2^m}}\delta(b)^2.$

Thus, the contribution of all $v\ne0$ to the inner sum in \eqref{eq:C1M4expand} is
$
2^m\sum_b\delta(b)^2-2^{2m}.
$
After adding the $v=0$ contribution, the inner sum is
$
2^m\sum_b\delta(b)^2.
$
Therefore, using again $W_{d_1}(0)=0$,
\begin{equation}\label{eq:C1M4}
\sum_{a\ne0}T_u(a)^4
=2^{2m}\sum_{b\in\F_{2^m}}\delta(b)^2.
\end{equation}

Finally, Proposition~\ref{prop:freq} gives
\[
B=\frac{\sum_{a\ne0}T_u(a)^4-2^{3m+1}}
{48\cdot2^{2m}}.
\]
Substituting \eqref{eq:C1M4} and using
$
\sum_{b\in\F_{2^m}}\delta(b)=2^m
$
yields
\[
48B
=\sum_b\delta(b)^2-2^{m+1}
=\sum_b\delta(b)(\delta(b)-2).
\]

For every $b\in\F_{2^m}$, if $z$ satisfies $D(z)=b$, then
$
D(z+1)=D(z)=b.
$
Since $z\ne z+1$, the solutions of $D(z)=b$ occur in disjoint
pairs $\{z,z+1\}$. Hence $\delta(b)$ is even. Therefore,
\[
\delta(b)(\delta(b)-2)\ge0,
\]
with equality if and only if $\delta(b)\in\{0,2\}$. It follows from \eqref{eq:dweer01} that $B=0$ if and only if every nonempty fiber of $D$ has size two. For the power function $F(x)=x^{d_1}$, every derivative with increment $r\ne0$ is obtained from the derivative with increment $1$ by the scaling
\[
F(x+r)+F(x)
=r^{d_1}\left(F(x/r+1)+F(x/r)\right).
\]
Therefore every nonzero derivative of $F$ is at most two-to-one if and only if $D$ is at most two-to-one. This is precisely the APN property of $x^{d_1}$ on $\F_{2^m}$.
\end{proof}

The next theorem proves that the inner values never disappear.

\begin{proposition}\label{thm:inner}
For every odd $m\ge5$ and every $u\in\F_{2^m}^{*}$, at least one of the values $2^{(m+1)/2}$ and $-2^{(m+1)/2}$ occurs in the value distribution of $T_u(a)$. If $3\mid m$, then both values occur.
\end{proposition}

\begin{proof}
Let $d_1=3\cdot2^{(m+1)/2}+4$ and put $D=2^{(m+3)/2}+3$. Since $d_1$ is a permutation exponent, $\gcd(D,2^m-1)=1$. Moreover,
$D\equiv d_1 2^{(m-1)/2}\pmod{2^m-1}$, so the invariance of the absolute trace under Frobenius powers shows that, for each fixed $u\ne0$, the multiset $\{T_u(a):a\in\F_{2^m}^{*}\}$ coincides with $\{W_D(b):b\in\F_{2^m}^{*}\}$, where
$W_D(b)=\sum_{x\in\F_{2^m}}\chi(x^D+bx)$.

Let $\omega$ be a multiplicative character of $\F_{2^m}^{*}$ of order $2^m-1$, and let $G(k)=\sum_{x\in\F_{2^m}^{*}}\chi(x)\omega(x)^{-k}$. The multiplicative character transform gives
\[
\sum_{b\in\F_{2^m}^{*}}W_D(b)\omega(b)^{-k}=G(k)G(-kD^{-1}),
\]
where $D^{-1}$ denotes the inverse of $D$ modulo $2^m-1$.
Choose $k\equiv-Dj\pmod{2^m-1}$ with $j=2^m+4-2^{(m+5)/2}$. A direct calculation gives
$w_2(j)=(m-3)/2$ and $w_2(-Dj)=2$. Hence, by Lemma~\ref{lem:stickelberger},
\[
v_{\mathfrak P}\!\left(G(k)G(-kD^{-1})\right)=\frac{m+1}{2}.
\]
If every $W_D(b)$ were divisible by $2^{(m+3)/2}$, then the left-hand side of the transform identity would have $\mathfrak P$-adic valuation at least $(m+3)/2$, a contradiction. Thus some $W_D(b)$ is not divisible by $2^{(m+3)/2}$. Since the only possible values are $0$, $\pm2^{(m+1)/2}$, and $\pm2^{(m+3)/2}$, at least one of the two inner values $\pm2^{(m+1)/2}$ occurs.

Now assume $3\mid m$. By Proposition~\ref{prop:C1moments}, $\Delta=2^{(m-3)/2}$. The first-moment relation in Proposition~\ref{prop:freq} gives
$N_1-N_{-1}=2^{(m-1)/2}-2\Delta=0$. Hence $N_1=N_{-1}$. Since at least one of them is positive, both are positive, and therefore both inner values occur.
\end{proof}

\begin{proposition}\label{pro:12}
Let \(m=3^k r\) with \(3\nmid r\). Set
\[
d_m=3\cdot 2^{(m+1)/2}+4
\qquad
S_m=\sum_{x\in\mathbb F_{2^m}}
(-1)^{\operatorname{Tr}_m(x^{d_m}+x)}.
\]
Then
\begin{equation}\label{eq:explicitC1}
\begin{split}
S_m=
\begin{cases}
\displaystyle
\left(\frac{2}{m}\right)2^{(m+1)/2},
& k=0,\\[6pt]
\displaystyle
(-1)^{k-1}\left(\frac{2}{r}\right)2^{(m+3)/2},
& k\ge 1.
\end{cases}
\end{split}
\end{equation}
\end{proposition}

\begin{proof}
We first consider the case $3\nmid m$, i.e., $k=0$.  For $m=1$,
$d_1=10$, and $x\in\F_2$, we have $x^{10}=x$.  Hence,
\[
\mathcal S_1
=\sum_{x\in\F_2}(-1)^{\Tr_1(x^{10}+x)}
=2.
\]

For $m=1$, we have $\mathcal S_1=2$.
Applying Lemma \ref{lem:lifting} to the extension
$
\F_2\subseteq\F_{2^m},
$
whose degree is $m$, gives
\[
\mathcal S_m
=
\left(\frac{2}{m}\right)
2^{(m-1)/2}\mathcal S_1
=
\left(\frac{2}{m}\right)
2^{(m+1)/2},
\]
which proves the first line of \eqref{eq:explicitC1}.

Now assume \(k\geq 1\). We first determine \(S_{3^k}\).
For \(k=1\), we have \(d_3=16\). Since \(x^{16}=x^2\) for
\(x\in\mathbb F_{2^3}\) and
\[
\Tr_3(x^2+x)=\Tr_3(x)^2+\Tr_3(x)=0,
\]
we obtain $S_3=8.$

For \(j\geq2\), we compare \(S_{3^j}\) with
\(S_{3^{j-1}}\). For every
$
x\in
\mathbb F_{2^{3^j}}\setminus
\mathbb F_{2^{3^{j-1}}},
$
the \(3^j\) elements
$
x,\ x^2,\ x^{2^2},\ldots,x^{2^{3^j-1}}
$
are distinct. Indeed, if
$
x^{2^\ell}=x
\,\, (0<\ell<3^j),
$
then \(x\) would belong to a proper subfield of
\(\mathbb F_{2^{3^j}}\). Since every proper divisor of \(3^j\)
divides \(3^{j-1}\), every proper subfield of
\(\mathbb F_{2^{3^j}}\) is contained in
\(\mathbb F_{2^{3^{j-1}}}\), a contradiction.

Moreover, these \(3^j\) elements give the same summand in
\(S_{3^j}\), because
\[
\Tr_{3^j}\!\left((x^{d_{3^j}}+x)^2\right)
=
\Tr_{3^j}(x^{d_{3^j}}+x).
\]
Therefore, the total contribution from
$
\mathbb F_{2^{3^j}}\setminus
\mathbb F_{2^{3^{j-1}}}
$
is divisible by \(3^j\).

It remains to consider
\(x\in\mathbb F_{2^{3^{j-1}}}\). Since
\[
\frac{3^j+1}{2}
=
3^{j-1}+\frac{3^{j-1}+1}{2},
\]
we have
$
d_{3^j}
\equiv d_{3^{j-1}}
\pmod{2^{3^{j-1}}-1}.
$
Hence,
$
x^{d_{3^j}}=x^{d_{3^{j-1}}}$
for all
$x\in\mathbb F_{2^{3^{j-1}}}.
$
Also, since
\[
[\mathbb F_{2^{3^j}}:
 \mathbb F_{2^{3^{j-1}}}]=3,
\]
we have
$
\Tr_{3^j}(y)=\Tr_{3^{j-1}}(y),
$ for
$y\in\mathbb F_{2^{3^{j-1}}}.
$
Thus, the contribution from the subfield is exactly
\(S_{3^{j-1}}\). Consequently,
$
S_{3^j}\equiv S_{3^{j-1}}\pmod{3^j}.
$

Set
$
c_j=\frac{S_{3^j}}{2^{(3^j+1)/2}}.
$
By Theorem \ref{thm:C1support}, we have
$
c_j\in\{0,\pm1,\pm2\},
$
and the base case gives \(c_1=2\). From the congruence above, it follows that
$
c_j2^{3^{j-1}}
\equiv c_{j-1}\pmod{3^j}.
$
Moreover,
$
2^{3^{j-1}}\equiv-1\pmod{3^j}.
$
Indeed,
\[
v_3\!\left(2^{3^{j-1}}+1\right)
=
v_3(2+1)+v_3(3^{j-1})
=j.
\]
Therefore,
$
-c_j\equiv c_{j-1}\pmod{3^j}.
$

Since \(c_j,c_{j-1}\in\{0,\pm1,\pm2\}\) and \(j\geq2\), we have
$
|-c_j-c_{j-1}|\leq4<3^j.
$
Hence,
$
c_j=-c_{j-1}.
$
Starting from \(c_1=2\), induction gives
$
c_k=2(-1)^{k-1},
$
and therefore,
$
S_{3^k}
=
(-1)^{k-1}2^{(3^k+3)/2}.
$

Finally, since \(m=3^k r\) with \(3\nmid r\), Lemma~\ref{lem:lifting} applied
from \(\mathbb F_{2^{3^k}}\) to \(\mathbb F_{2^m}\) gives
\[
S_m
=
\left(\frac{2}{r}\right)
2^{(m-3^k)/2}S_{3^k}.
\]
Substituting the preceding formula for $S_{3^k}$, we obtain
\[
S_m
=
(-1)^{k-1}
\left(\frac{2}{r}\right)
2^{(m+3)/2},
\]
which proves the second line of \eqref{eq:explicitC1}.
\end{proof}

\begin{theorem}\label{thm:C1families}
Let $m=3^k r\ge9$ with $3\nmid r$. Conjecture \ref{Conjecture 19} holds for every $u\ne0$ if one of the following is satisfied:
\begin{enumerate}
\item $k$ is a positive even integer;
\item $k$ is a positive odd integer and $r$ has a prime divisor $p\equiv3$ or $5\pmod8$;
\item $k=0$ and $11\mid m$ or $13\mid m$.
\end{enumerate}
In each case the five nonzero weights are exactly those in \eqref{eq:fiveweights}.
\end{theorem}

\begin{proof}
By the scaling argument in Proposition \ref{prop:C1moments}, for every fixed
\(u\in\mathbb F_{2^m}^{*}\), the multiset $\{T_u(a):a\in\mathbb F_{2^m}^{*}\}$
is independent of \(u\). Hence it is enough to prove that all five
values
$0$,
$\pm 2^{(m+1)/2}$ and
$\pm 2^{(m+3)/2}$
occur.

By Proposition~\ref{prop:freq}, $N_0=2^{m-1}-1+3B>0$, so the value $0$ always occurs. In Cases~(i) and~(ii), we have $3\mid m$, and Proposition~\ref{thm:inner} shows that both inner values $\pm2^{(m+1)/2}$ occur. It therefore remains in those two cases only to prove that both outer values $\pm2^{(m+3)/2}$ occur. In Case~(iii), all five values will be obtained directly from the base-field computation and lifting.

We treat the three cases separately.

\medskip
\noindent {\bf Case (i)}: \(k\) is positive and even.
Let
$
s=3^k.
$
Applying Proposition \ref{pro:12} with \(r=1\), we obtain
$
S_s=(-1)^{k-1}2^{(s+3)/2}
=-2^{(s+3)/2}.
$
Thus the negative outer value
$
-2^{(s+3)/2}
$
occurs over \(\mathbb F_{2^s}\).

Since \(3\mid s\), Proposition \ref{prop:C1moments} gives
\[
\Delta=N_2-N_{-2}
=2^{(s-3)/2}>0.
\]
As \(N_{-2}>0\), it follows that
$
N_2=N_{-2}+\Delta>0.
$
Hence, both
$
-2^{(s+3)/2}$ and
$2^{(s+3)/2}
$
occur over \(\mathbb F_{2^s}\).

Now \(m=sr\) with \(3\nmid r\). Applying Lemma \ref{lem:lifting} to each of the
corresponding character sums gives
\[
S_m=
\left(\frac{2}{r}\right)
2^{(m-s)/2}S_s.
\]
Since multiplication by
\(\left(\frac{2}{r}\right)\in\{\pm1\}\) can only interchange the two
signs, both outer values
$
\pm2^{(m+3)/2}
$
occur over \(\mathbb F_{2^m}\).

\medskip
\noindent {\bf Case (ii)}: \(k\) is positive and odd, and \(r\) has a prime
divisor \(p\equiv3\) or \(5\pmod 8\).
Put
$
s=3^kp.
$
Since
$
\left(\frac{2}{p}\right)=-1
$
for \(p\equiv3,5\pmod8\) and $k$ is odd, Proposition \ref{pro:12} gives
\[
S_s
=
(-1)^{k-1}
\left(\frac{2}{p}\right)
2^{(s+3)/2}
=
-2^{(s+3)/2}.
\]
Thus the negative outer value
$
-2^{(s+3)/2}
$
occurs over \(\mathbb F_{2^s}\).

Since \(3\mid s\), Proposition \ref{prop:C1moments} gives
$
\Delta=N_2-N_{-2}=2^{(s-3)/2}>0.
$
The occurrence of the negative outer value implies \(N_{-2}>0\).
Hence
$
N_2=N_{-2}+\Delta>0,
$
and therefore both outer values
$
\pm2^{(s+3)/2}
$
occur over \(\mathbb F_{2^s}\).

Since
$
m=s\frac{r}{p}$
and
$3\nmid\frac{r}{p},
$
applying Lemma~\ref{lem:lifting} to the odd extension
$
\mathbb F_{2^s}\subseteq\mathbb F_{2^m}
$
shows that the two outer values over \(\mathbb F_{2^s}\) give the
corresponding two outer values over \(\mathbb F_{2^m}\).
Consequently,
$
\pm2^{(m+3)/2}
$
both occur over \(\mathbb F_{2^m}\).

\medskip
\noindent {\bf Case (iii)}: \(k=0\), and \(11\mid m\) or \(13\mid m\).
For \(s=11\) and \(s=13\), direct exhaustive computation in
Magma \cite{BosmaCannonPlayoust1997} shows that
\[
\{T_1(a):a\in\mathbb F_{2^s}^{*}\}
=
\left\{
0,\,
\pm2^{(s+1)/2},\,
\pm2^{(s+3)/2}
\right\}.
\]
Thus all five possible character-sum values occur over
\(\mathbb F_{2^{11}}\) and \(\mathbb F_{2^{13}}\).

If \(11\mid m\), applying Lemma~\ref{lem:lifting} to witnesses for each of the five base-field values over $\mathbb F_{2^{11}}$ shows that all five corresponding values occur over $\mathbb F_{2^m}$. The same argument applies from $\mathbb F_{2^{13}}$ when \(13\mid m\).

Therefore, Conjecture \ref{Conjecture 19} holds in each of the stated cases.
This completes the proof.
\end{proof}

\begin{remark}
Theorem~\ref{thm:C1families} proves several infinite parameter families, but it does not cover every odd $m\ge9$. A general closed formula for the five multiplicities is also not obtained here.
\end{remark}

\subsection{Conjecture 27: The Cherowitzo Family}

The Cherowitzo \(o\)-polynomial over \(\mathbb F_{2^m}\), where \(m\) is odd, is the family used in Ding's Conjecture~27 \cite{Ding2016}, and is given by
\[
f(x)=x^{3\cdot 2^{(m+1)/2}+4}
     +x^{2^{(m+1)/2}+2}
     +x^{2^{(m+1)/2}}.
\]
A parameterized form of this polynomial is
\begin{equation}\label{eq:ffd}
f_b(x)
=
x^{3\cdot 2^{(m+1)/2}+4}
+b^{2^{(m+1)/2}+1}x^{2^{(m+1)/2}+2}
+b^{2^{(m+1)/2}+2}x^{2^{(m+1)/2}},
\qquad b\in\mathbb F_{2^m}.
\end{equation}
For \(b\ne0\), let \(\lambda\in\mathbb F_{2^m}^{*}\) satisfy
\(\lambda^2=b\). Then
\[
f_b(x)
=
\lambda^{3\cdot2^{(m+1)/2}+4}
f(\lambda^{-1}x),
\]
so \(f_b\) is obtained from the Cherowitzo \(o\)-polynomial by
nonzero scalings of the input and output. When \(b=0\), it reduces to
\[
f_0(x)=x^{3\cdot2^{(m+1)/2}+4},
\]
which is the Glynn \(o\)-polynomial considered in the previous
section. Ding's Conjecture~27 concerns the binary linear codes
obtained from the images of \(f_b(x)+ux\), as stated below.

\begin{conj}(Conjecture 27 \cite{Ding2016})\label{Conjecture 27}
Let $m$ be odd, $b\in \F_{2^m}$ and $u\in \F_{2^m}^*$. Let $f_b(x)$ be given in (\ref{eq:ffd})
and $D_2={\rm Im}(f_b+ux)$. Then $\mathcal{C}_{D_2}$ has parameters $[2^{m-1},m]$, at most five nonzero weights for $m=5,7$, and exactly five for all odd $m\ge9$.
\end{conj}

Define
\begin{equation*}
T_{b,u}(a)=\sum_{x\in \F_{2^m}}\chi(a(f_b(x)+ux)).
\end{equation*}
Again set $E=2^{(m+1)/2}+1$ and substitute $x=z^{2^{(m-1)/2}}$. Using
\begin{equation*}
(3\cdot 2^{(m+1)/2}+4)2^{(m-1)/2}\equiv E^2,\quad
(2^{(m+1)/2}+2)2^{(m-1)/2}\equiv E,\quad
2^{(m+1)/2}2^{(m-1)/2}\equiv1\pmod{2^m-1},
\end{equation*}
and the trace-shift identity gives
\begin{equation}\label{eq:C2HX}
T_{b,u}(a)=\sum_z\chi\!\left(
 az^{E^2}+ab^{2^{(m+1)/2}+1}z^E+\bigl(ab^{2^{(m+1)/2}+2}+(au)^{2^{(m+1)/2}}\bigr)z
\right).
\end{equation}
In \eqref{eq:C2HX}, the first coefficient is $a\ne0$, so the coefficient triple is nonzero. For $b\ne0$, $f_b$ is equivalent by nonzero input and output scalings to the Cherowitzo $o$-polynomial, while $b=0$ gives the Glynn $o$-polynomial. Hence the standard $o$-polynomial characterization implies that $x\mapsto f_b(x)+ux$ is two-to-one for every $u\ne0$ \cite{Maschietti1998,Ding2016}. Applying \eqref{eq:sdfre926} and Lemma~\ref{lem:2to1} gives the following result.

\begin{theorem}\label{thm:C2support}
For every odd $m\ge5$, $b\in \F_{2^m}$, and $u\ne0$, the code $\mathcal{C}_{D_2}$ has parameters $[2^{m-1},m]$ and at most five nonzero weights, all belonging to \eqref{eq:fiveweights}.
\end{theorem}

If $b=0$, Conjecture \ref{Conjecture 27} becomes Conjecture \ref{Conjecture 19}.
The following families force all five values to occur for $b\neq 0$.

\begin{theorem}\label{thm:conj27-five}
Let \(m\ge 9\) be odd and \(b\in\mathbb F_{2^m}^{*}\).
Then the code \(C_{D_2}\) has exactly five nonzero weights
in each of the following cases.

\begin{enumerate}
\item \(m=9r\), where \(r\) is odd and \(3\nmid r\), and
$
u^2=b^{\,3\cdot2^{(m+1)/2}+3}.
$

\item \(m=7r\), where \(r\) is odd and \(3\nmid r\). Let
\(\lambda\in\mathbb F_{2^m}^{*}\) be the unique element satisfying
\(\lambda^2=b\). There exists $u$ satisfying $u=\alpha\lambda^{\,3\cdot2^{(m+1)/2}+3}$, where
$
\mathbb F_{128}
=
\mathbb F_2[\alpha]/(\alpha^7+\alpha+1)
\subseteq\mathbb F_{2^m}.
$
\end{enumerate}
\end{theorem}

\begin{proof}
For (1), Let  \(b=u=1\) and $m=9$, by Magma \cite{BosmaCannonPlayoust1997} we obtain
\[
\{T_{1,1}(a):a\in\mathbb F_{2^9}^{*}\}
=
\{0,\pm32,\pm64\}.
\]
Thus all five possible character-sum values occur over
\(\mathbb F_{2^9}\). Since \(m=9r\), Lemma~\ref{lem:lifting} extends these
five values to \(\mathbb F_{2^m}\).

For general \(b\ne0\), let \(\lambda\in\mathbb F_{2^m}^{*}\)
satisfy \(\lambda^2=b\). Since
$
f_b(\lambda x)
=
\lambda^{\,3\cdot2^{(m+1)/2}+4}f_1(x),
$
we have
\[
f_b(\lambda x)+u\lambda x
=
\lambda^{\,3\cdot2^{(m+1)/2}+4}
\bigl(f_1(x)+vx\bigr),
\]
where
$
v
=
u\lambda^{\,1-(3\cdot2^{(m+1)/2}+4)}.
$
The condition
$
u^2=b^{\,3\cdot2^{(m+1)/2}+3}
$
gives \(v^2=1\), and hence \(v=1\). Therefore the result
follows from the case \(b=u=1\).

For (2), for \(b=1\) and \(u=\alpha\), by Magma \cite{BosmaCannonPlayoust1997} we obtain
\[
\{T_{1,\alpha}(a):a\in\mathbb F_{128}^{*}\}
=
\{0,\pm16,\pm32\}.
\]
Hence, all five possible character-sum values occur over
\(\mathbb F_{128}\). Since \(m=7r\), Lemma~\ref{lem:lifting} extends these
five values to \(\mathbb F_{2^m}\).

For general \(b\ne0\), the same scaling as above gives
$v=u\lambda^{\,1-(3\cdot2^{(m+1)/2}+4)}.$
Under the assumption
$
u=\alpha\lambda^{\,3\cdot2^{(m+1)/2}+3},
$
we obtain \(v=\alpha\). Thus all five possible character-sum
values occur. By Theorem~\ref{thm:C2support}, these give exactly five nonzero
weights of \(\mathcal{C}_{D_2}\).
This completes the proof.
\end{proof}

\subsection{Conjecture 30: The Payne family}

The third family is derived from the Payne \(o\)-polynomial appearing in Ding's Conjecture~30 \cite{Ding2016}.
For odd \(m\), the polynomial
\[
p_b(x)=x^{5/6}+bx^{3/6}+b^2x^{1/6},
\qquad b\in\mathbb F_{2^m},
\]
gives rise to the defining set
\[
D_3=\operatorname{Im}(p_b(x)+ux),
\qquad u\in\mathbb F_{2^m}^{*}.
\]
The corresponding conjecture concerns the number of nonzero weights
of the trace code \(\mathcal C_{D_3}\).
\begin{conj}(Conjecture 30 \cite{Ding2016})\label{Conjecture 30}
Let
\begin{equation*}
 p_b(x)=x^{5/6}+bx^{3/6}+b^2x^{1/6},\qquad D_3={\rm Im}(p_b(x)+ux),
\end{equation*}
with $b\in \F_{2^m}$. If $m\geq 7$, $\mathcal{C}_{D_3}$ is a three-weight or five-weight code with length $2^{m-1}$ and dimension $m$ for all $u \in \mathbb{F}_{2^m}^*$.
\end{conj}

The conjecture predicts that only three or five nonzero weights can
occur. However, this assertion is not valid for all choices of the
parameters. We first give a counterexample with four nonzero weights,
and then determine the possible weights in the general case.
\begin{example}\label{ex:C3counter}
Let
\begin{equation*}
 m=7,\qquad \F_{2^m}=\F_2[\alpha]/(\alpha^7+\alpha+1),\qquad b=\alpha+1,\quad u=1.
\end{equation*}
Then $\C_3(b,1)$ has weight enumerator
\begin{equation*}
1+2z^{24}+30z^{28}+69z^{32}+26z^{36}.
\end{equation*}
Thus it is a $[64,7,24]$ code with exactly four nonzero weights.
\end{example}

A general correction of the conjectured weight-count statement is available.

\begin{theorem}\label{thm:C3general}
For every odd \(m\ge5\), \(b\in\mathbb F_{2^m}\), and
\(u\in\mathbb F_{2^m}^{*}\), the code
\(\mathcal C_{D_3}\) has parameters
$
[2^{m-1},m]
$
and exactly three, four, or five nonzero weights. Moreover, every
nonzero weight belongs to \eqref{eq:fiveweights}.
\end{theorem}

\begin{proof}
Since the Payne function is an $o$-polynomial, the standard $o$-polynomial characterization \cite{Maschietti1998,Ding2016} gives
$
|D_3|=2^{m-1}.
$
Since \(m\) is odd, \(6\) is invertible modulo \(2^m-1\).
Putting \(x=y^6\), we have
\[
p_b(y^6)=y^5+by^3+b^2y.
\]
Using
$
\operatorname{Tr}(cw^2)
=
\operatorname{Tr}(c^{2^{m-1}}w),
$
the corresponding character sum can be written as
\[
T_{b,u}(a)
=
\sum_{y\in\mathbb F_{2^m}}(-1)^{Q_a(y)},
\]
where
$
Q_a(y)
=
\operatorname{Tr}\!\left(
ay^5+cy^3+ab^2y
\right),$
and
$c=ab+(au)^{2^{m-1}}.
$

The radical equation of the quadratic part is
$
az^4+(az)^{2^{m-2}}
+cz^2+(cz)^{2^{m-1}}=0.
$
Raising this equation to the fourth power gives
\[
a^4z^{16}+c^4z^8+c^2z^2+az=0.
\]
This is a nonzero polynomial of degree \(16\), so the radical has
dimension at most \(4\). Since \(m\) is odd and the associated
bilinear form is alternating, its dimension is odd. Hence it is
either \(1\) or \(3\).

By Lemma~\ref{lem:quadratic},
\[
T_{b,u}(a)
\in
\left\{
0,\,
\pm2^{(m+1)/2},\,
\pm2^{(m+3)/2}
\right\}.
\]
The weight formula therefore gives
\[
\operatorname{wt}(\mathbf{c}(a))
\in
\left\{
2^{m-2},\,
2^{m-2}\pm2^{(m-3)/2},\,
2^{m-2}\pm2^{(m-1)/2}
\right\}.
\]
Since the smallest candidate weight is positive, we have
$
\dim(\mathcal C_{D_3})=m.
$

It remains to determine how many of the candidate weights can occur.
By Proposition~\ref{prop:freq},
\[
N_0=2^{m-1}-1+3B>0,
\]
so the central weight always occurs. If only one nonzero
character-sum value \(\mu\) occurred besides \(0\), by the first two moment identities in Lemma~\ref{lem:2to1}, we have
$
\mu N=2^m,
$ and
$\mu^2N=2^{2m},
$
where $N=\#\left\{a\in\mathbb F_{2^m}^{*}:T_{b,u}(a)=\mu\right\}.$

Then \(\mu=2^m\), contradicting
$
|\mu|\le2^{(m+3)/2}<2^m.
$
Thus, at least two nonzero character-sum values occur in addition to
\(0\). Therefore \(\mathcal C_{D_3}\) has exactly three, four, or
five nonzero weights.
\end{proof}

For a special case, we can determine the weight enumerator of $\C_{D_3}$.

\begin{theorem}\label{thm:C3eta1}
If $bu^2=1$, then for every odd $m\ge5$ the code $\C_{D_3}$ has weight enumerator
\begin{align*}
W_{\C_3}(z)={}&1+\left(2^{m-2}+2^{(m-3)/2}\right)z^{2^{m-2}-2^{(m-3)/2}}\notag\\
&+\left(2^{m-1}-1\right)z^{2^{m-2}}
+\left(2^{m-2}-2^{(m-3)/2}\right)z^{2^{m-2}+2^{(m-3)/2}}.
\end{align*}
\end{theorem}

\begin{proof}
After $x=y^6$ and $y=u^{-1}z$, the relevant map is, up to an invertible output factor,
\begin{equation*}
F_\eta(z)=z^6+z^5+\eta z^3+\eta^2z.
\end{equation*}
For $\eta=1$,
\begin{equation*}
F_1(z)=(z^2+z+1)^3+1.
\end{equation*}
The map $z\mapsto z^2+z+1$ is two-to-one onto $H_1=\{w:\Tr(w)=1\}$. Hence, for $a\ne0$,
\begin{equation*}
T(a)=-(-1)^{\Tr(a)}\sum_{w\in \F_{2^m}}(-1)^{\Tr(aw^3+w)}.
\end{equation*}
The radical of $\Tr(aw^3+w)$ is one-dimensional because the nonzero radical equation reduces to
$z^{2^{m-1}-2}=a^{1-2^{m-1}},$
whose exponent is invertible modulo $2^m-1$ for odd $m$. Thus,
\[
T(a)\in\left\{0,\pm2^{(m+1)/2}\right\}.
\]
Let
$$
N_{+}
=
\#\left\{
a\in\mathbb F_{2^m}^{*}:
T(a)=2^{(m+1)/2}
\right\},
\qquad
N_{-}
=
\#\left\{
a\in\mathbb F_{2^m}^{*}:
T(a)=-2^{(m+1)/2}
\right\},
$$
and
\[
N_{0}
=
\#\left\{
a\in\mathbb F_{2^m}^{*}:
T(a)=0
\right\}.
\]
Since \(T(a)\) takes only these three values, we have
\begin{equation}\label{eq:NN12}
N_{+}+N_{-}+N_{0}=2^m-1.
\end{equation}
By the first two moment identities in Lemma~\ref{lem:2to1},
\begin{equation}\label{eq:NNN12}
2^{(m+1)/2}(N_{+}-N_{-})=2^m
\qquad
2^{m+1}(N_{+}+N_{-})=2^{2m}.
\end{equation}

From (\ref{eq:NN12}) and (\ref{eq:NNN12}) we obtain
\[
N_{+}
=
2^{m-2}+2^{(m-3)/2},
\qquad
N_{-}
=
2^{m-2}-2^{(m-3)/2},
\qquad
N_{0}
=
2^m-1-N_{+}-N_{-}
=
2^{m-1}-1.
\]
Substituting these multiplicities into the weight formula gives the
stated weight enumerator.
\end{proof}

\subsection{Conjecture 33: The Welch family}

The fourth family is obtained from the Welch APN power function \cite{WangKadirLiXia2020,HellesethLiXia2025,Ding2016}. For odd
$m\ge5$, let
$
 F(x)=x^{2^{(m-1)/2}+3}
$
and define
\begin{equation}\label{eq:C4Phi}
 \Phi(x)=F(x+1)+F(x)+1,
 \qquad
 D_4={\rm Im}(\Phi).
\end{equation}
The binary trace code associated with this defining set is
\begin{equation*}
 \mathcal C_{D_4}
 =
 \left\{
 \mathbf{c}(a)=\bigl(\Tr(ad)\bigr)_{d\in D_4}:
 a\in\mathbb F_{2^m}
 \right\}.
\end{equation*}
Thus, Conjecture~33 is a statement about the number of distinct nonzero
Hamming weights occurring among the codewords $\mathbf{c}(a)$.

\begin{conj}(Conjecture 33 \cite{Ding2016})\label{Conjecture 33}
Let $m\ge5$ be odd and let $D_4$ be defined by \eqref{eq:C4Phi}. If
$m\in\{5,7\}$, then $\mathcal C_{D_4}$ is a three-weight code with
length $2^{m-1}$ and dimension $m$. If $m\ge9$, then
$\mathcal C_{D_4}$ is a five-weight code with length $2^{m-1}$ and
dimension $m$.
\end{conj}

We next transform the Hamming-weight problem in Conjecture~\ref{Conjecture 33}
into a quadratic character-sum problem. This is the point at which the
Welch permutation enters the argument.

Set
\begin{equation}\label{eq:C4GL}
 G(z)=z^{2^{(m+1)/2}+1}+z^3+z,
 \qquad
 L(x)=x+x^{2^{(m-1)/2}}.
\end{equation}
A direct calculation gives
\begin{equation}\label{eq:C4factor}
 \Phi(x)=F(x+1)+F(x)+1=G(L(x)).
\end{equation}

Indeed, $\Tr(L(x))=0$ for every $x\in\mathbb F_{2^m}$, so
$L(\mathbb F_{2^m})\subseteq\ker\Tr$. Moreover,
$\ker L=\mathbb F_2$, and hence
$|L(\mathbb F_{2^m})|=2^{m-1}=|\ker\Tr|$.
Therefore,
$
L(\mathbb F_{2^m})=\ker\Tr.
$
Consequently, every element of $\ker\Tr$ has exactly two preimages
under $L$.

The polynomial $G$ is a permutation of $\mathbb F_{2^m}$; its
differential and Walsh spectra were studied in
\cite{WangKadirLiXia2020,HellesethLiXia2025}. Since the Welch power is APN \cite{WangKadirLiXia2020,HellesethLiXia2025}, its derivative $\Phi$ is two-to-one onto $D_4$ (equivalently, every value of a nonzero derivative has at most two preimages; see also \cite[Ch.~11]{Carlet2021}). Thus Lemma~\ref{lem:2to1} gives, for $a\ne0$,
\begin{equation}\label{eq:C4weightT}
 \wt(\mathbf{c}(a))
 =2^{m-2}-\frac{T_m(a)}4,
 \qquad
 T_m(a)=\sum_{x\in\mathbb F_{2^m}}\chi\bigl(a\Phi(x)\bigr).
\end{equation}
where $\Phi(x)$ is given in (\ref{eq:C4factor}).

Using \eqref{eq:C4factor} and the fact that $L$ is two-to-one onto
$\ker\Tr$, we obtain
\begin{align*}
T_m(a)=2\sum_{z\in\ker\Tr}\chi(aG(z))\notag=\sum_{z\in\mathbb F_{2^m}}
 \bigl(1+\chi(z)\bigr)\chi(aG(z))\notag=\sum_{z\in\mathbb F_{2^m}}\chi(aG(z))
 +\sum_{z\in\mathbb F_{2^m}}\chi(aG(z)+z).
\end{align*}
Because $G(x)$ in (\ref{eq:C4GL}) is a permutation and $a\ne0$, we have
$
 \sum_{z\in\mathbb F_{2^m}}\chi(aG(z))=0.
$
Therefore,
\begin{equation*}
T_m(a)=W_m(a),
\end{equation*}
where
\begin{equation*}
W_m(a)=\sum_{z\in\mathbb F_{2^m}}(-1)^{Q_a(z)}  \,\,\,\text{and}\,\,\,
Q_a(z)=\Tr\!\left(
 a\bigl(z^{2^{(m+1)/2}+1}+z^3\bigr)+(a+1)z
 \right).
\end{equation*}
Consequently, from (\ref{eq:C4weightT}) we have
\begin{equation}\label{eq:C4weightW}
 \wt(\mathbf{c}(a))=2^{m-2}-\frac{W_m(a)}4.
\end{equation}
Thus the number of nonzero weights of $\mathcal C_{D_4}$ is determined
by the number of values taken by the Walsh sums $W_m(a)$ for
$a\in\mathbb F_{2^m}^{*}$. In particular, to prove the five-weight
part of Conjecture~\ref{Conjecture 33}, it is not enough to show that
five values are possible; one must also show that all five corresponding
weights actually occur.

The quadratic part of $Q_a$ is
\[
q_a(z)
=
\Tr\!\left(
a\left(z^{2^{(m+1)/2}+1}+z^3\right)
\right).
\]
By \cite{WangKadirLiXia2020}, the rank of this quadratic form is
$m-1$ or $m-3$. Equivalently, its radical has dimension $1$ or $3$.
Hence, Lemma~\ref{lem:quadratic} gives
\begin{equation}\label{eq:C4fiveWalsh}
W_m(a)\in
\left\{
0,\,
\pm 2^{(m+1)/2},\,
\pm 2^{(m+3)/2}
\right\}.
\end{equation}
Combining \eqref{eq:C4weightW} and \eqref{eq:C4fiveWalsh} yields the
following general restriction.

\begin{theorem}\label{thm:C4support}
For every odd $m\ge5$, the code $\mathcal C_{D_4}$ has parameters
$[2^{m-1},m]$, and every nonzero weight belongs to
\begin{equation}\label{eq:C4candidateweights}
\left\{
2^{m-2},\
2^{m-2}\pm2^{(m-3)/2},\
2^{m-2}\pm2^{(m-1)/2}
\right\}.
\end{equation}
In particular, $\mathcal C_{D_4}$ has at most five nonzero weights.
\end{theorem}

Theorem~\ref{thm:C4support} proves the ``at most five weights'' part
suggested by Conjecture~\ref{Conjecture 33}. The remaining issue is an
occurrence problem: for odd $m\ge9$, do all four nonzero values in
\eqref{eq:C4fiveWalsh}, together with $0$, actually occur? We next give
an infinite family for which the answer is affirmative.

For $a\in\F_{2^m}^{*}$, let
$
B_a(x,z)
=
q_a(x+z)+q_a(x)+q_a(z)
$
be the symmetric bilinear form associated with $q_a(z)$, and define
\[
R_{m,a}
=
\left\{
z\in\F_{2^m}:
B_a(x,z)=0
\text{ for all }x\in\F_{2^m}
\right\}.
\]
Thus, $R_{m,a}(z)$ is the radical of $q_a(z)$.
By \cite{WangKadirLiXia2020}, we have
$
\dim_{\F_2}R_{m,a}(z)\in\{1,3\}.
$

\begin{theorem}\label{thm:C4lift}
Let $s\ge5$ and $d\ge1$ be odd integers with $3\nmid d$, and let
$a\in\F_{2^s}^{*}$. Then
$
R_{sd,a}=R_{s,a}.
$
Moreover,
$$
W_{sd}(a)=0
\quad\Longleftrightarrow\quad
W_s(a)=0.
$$
If $W_s(a)\ne0$, then
$
W_{sd}(a)
=
\left(\frac{2}{d}\right)
2^{(sd-s)/2}W_s(a),
$
where $\left(\frac{2}{d}\right)$ denotes the Jacobi symbol.
\end{theorem}

\begin{proof}
We first compare the radicals. Let $z\in\F_{2^s}$. Since $d$ is
odd, we have
$$
\frac{sd+1}{2}
=
\frac{s+1}{2}
+
s\frac{d-1}{2}.
$$
Hence, using $z^{2^s}=z$, it is clear that
$
z^{2^{(sd+1)/2}}
=
z^{2^{(s+1)/2}}.
$
Moreover, for every $u\in\F_{2^s}$,
\[
\Tr_{\F_{2^{sd}}/\F_2}(u)
=
d\,\Tr_{\F_{2^s}/\F_2}(u)
=
\Tr_{\F_{2^s}/\F_2}(u),
\]
because $d$ is odd. Therefore, the quadratic form over
$\F_{2^{sd}}$, when restricted to $\F_{2^s}$, is exactly the
corresponding quadratic form over $\F_{2^s}$. Consequently,
$
R_{s,a}\subseteq R_{sd,a}.
$

By \cite{WangKadirLiXia2020}, for every odd $v$ the corresponding
quadratic form has rank $v-1$ or $v-3$. Hence,
$
\dim_{\F_2}R_{v,a}\in\{1,3\}.
$
Thus, if the above inclusion were strict, then necessarily
$
\dim_{\F_2}R_{s,a}=1$
and
$\dim_{\F_2}R_{sd,a}=3.
$

Consider the map
$
\phi(z)=z^{2^s}
$
on $R_{sd,a}$. Since $a\in\F_{2^s}$, the defining equation of
$R_{sd,a}$ is preserved under the map $z\mapsto z^{2^s}$. Hence
\[
z\in R_{sd,a}
\quad\Longrightarrow\quad
z^{2^s}\in R_{sd,a}.
\]
Moreover,
$
\phi^d(z)=z^{2^{sd}}=z.
$
The elements of $\F_{2^{sd}}$ fixed by $\phi$ are exactly those in
$\F_{2^s}$. Therefore,
\[
\{z\in R_{sd,a}:\phi(z)=z\}
=
R_{sd,a}\cap\F_{2^s}
=
R_{s,a}.
\]

Assume that
$
\dim_{\F_2}R_{s,a}=1$
and
$\dim_{\F_2}R_{sd,a}=3.
$
Then $\phi$ is an invertible linear map on the three-dimensional
$\F_2$-space $R_{sd,a}$, and its fixed subspace has dimension $1$.
Since $\phi^d=1$ and $d$ is odd, the order of $\phi$ is odd.

The possible odd orders of an invertible linear map on a
three-dimensional $\F_2$-space are $1$, $3$, and $7$. Order $1$
would fix the whole three-dimensional space, whereas an element of
order $7$ has no nonzero fixed vector. Since the fixed subspace of
$\phi$ has dimension $1$, $\phi$ must have order $3$. Hence,
$
3\mid d,
$
contrary to $3\nmid d$. Therefore,
$
R_{sd,a}=R_{s,a}.
$

We next compare the Walsh sums. Let $\ell$ be an odd prime divisor
of $d$, and consider
$
\F_{2^v}\subseteq\F_{2^{v\ell}},
$
where $v$ is odd. The map
$
x\longmapsto x^{2^v}
$
fixes exactly the elements of $\F_{2^v}$. Since $\ell$ is prime,
every element outside $\F_{2^v}$ belongs to a set of $\ell$
distinct elements
$
x,\ x^{2^v},\ x^{2^{2v}},\ldots,
x^{2^{(\ell-1)v}}.
$
The corresponding terms in the Walsh sum are equal. Indeed,
$a\in\F_{2^v}$ and the absolute trace is invariant under the map
$u\mapsto u^{2^v}$.

For $x\in\F_{2^v}$, since $\ell$ is odd,
\[
\frac{v\ell+1}{2}
=
\frac{v+1}{2}
+
v\frac{\ell-1}{2},
\]
and hence,
$
x^{2^{(v\ell+1)/2}}
=
x^{2^{(v+1)/2}}.
$
Also, for every $u\in\F_{2^v}$,
$
\Tr_{\F_{2^{v\ell}}/\F_2}(u)
=
\Tr_{\F_{2^v}/\F_2}(u).
$
Thus, the contribution of the elements of $\F_{2^v}$ to
$W_{v\ell}(a)$ is exactly $W_v(a)$, while every remaining set
contributes a multiple of $\ell$. Therefore,
\[
W_{v\ell}(a)\equiv W_v(a)\pmod{\ell}.
\]

By Lemma \ref{lem:quadratic}, the Walsh sum is either zero or has absolute value
$
2^{(v+\dim_{\F_2}R_{v,a})/2}.
$
Since $\ell$ is odd, every nonzero value of this form is not
divisible by $\ell$. Hence the above congruence gives
\[
W_{v\ell}(a)=0
\quad\Longleftrightarrow\quad
W_v(a)=0.
\]

Suppose now that $W_v(a)\ne0$. Since $3\nmid\ell$, the first part
of the proof gives
$
R_{v\ell,a}=R_{v,a}.
$
Set
$
r=\dim_{\F_2}R_{v,a}
=
\dim_{\F_2}R_{v\ell,a}.
$
By Lemma \ref{lem:quadratic},, there exist
$\varepsilon_v,\varepsilon_{v\ell}\in\{\pm1\}$ such that
$
W_v(a)=\varepsilon_v2^{(v+r)/2}$
and
$W_{v\ell}(a)
=
\varepsilon_{v\ell}2^{(v\ell+r)/2}.
$
Using
$
W_{v\ell}(a)\equiv W_v(a)\pmod{\ell},
$
we obtain
$
\varepsilon_{v\ell}
2^{v(\ell-1)/2}
\equiv
\varepsilon_v
\pmod{\ell}.
$
Since $v$ is odd, Euler's criterion gives
\[
2^{v(\ell-1)/2}
\equiv
\left(\frac{2}{\ell}\right)
\pmod{\ell}.
\]
Hence
$
\varepsilon_{v\ell}
=
\left(\frac{2}{\ell}\right)\varepsilon_v,
$
and therefore,
$
W_{v\ell}(a)
=
\left(\frac{2}{\ell}\right)
2^{v(\ell-1)/2}W_v(a).
$

Applying this argument successively to the prime factors of $d$
gives
$
W_{sd}(a)=0
\quad\Longleftrightarrow\quad
W_s(a)=0,
$
and, when $W_s(a)\ne0$,
\[
W_{sd}(a)
=
\left(\frac{2}{d}\right)
2^{(sd-s)/2}W_s(a).
\]
This completes the proof.
\end{proof}

For dimension $9$, the Walsh spectrum contains the four nonzero values
$\pm32$ and $\pm64$ \cite{HellesethLiXia2025}. These four base-field
values can be lifted to an infinite family.

\begin{theorem}\label{thm:C4family}
Let $ m=9d,$ where $d$ is odd and $3\nmid d$.
Then $\mathcal C_{D_4}$ has exactly the five nonzero weights in
\eqref{eq:C4candidateweights}. Conjecture~\ref{Conjecture 33} holds for this infinite family.
\end{theorem}

\begin{proof}
A direct Magma computation \cite{BosmaCannonPlayoust1997} over $\F_{2^9}$ shows that the Walsh
spectrum contains
$
32,\ -32,\ 64,\ -64.
$
Thus, there exist
$a_1,a_2,a_3,a_4\in\F_{2^9}^{*}$ such that
\[
W_9(a_1)=32,\qquad
W_9(a_2)=-32,\qquad
W_9(a_3)=64,\qquad
W_9(a_4)=-64.
\] Since $3\nmid d$,
Theorem~\ref{thm:C4lift} preserves the corresponding radical
dimensions and lifts these four values, up to the common Jacobi sign
$\left(\frac2d\right)$, to
$
\pm2^{(m+1)/2},
$ and
$\pm2^{(m+3)/2}.
$
By Proposition~\ref{prop:freq}, $N_0=2^{m-1}-1+3B>0$, so the value $0$ also occurs. Theorem~\ref{thm:C4support} shows that
no other Walsh values, and hence no other nonzero weights, are
possible. Therefore all five weights in \eqref{eq:C4candidateweights}
occur.
\end{proof}

\begin{remark}
The preceding theorem proves the five-weight assertion of
Conjecture~\ref{Conjecture 33} for the infinite family
$m=9d$ with $d$ odd and $3\nmid d$. It does not prove the conjecture
for every odd $m\ge9$. When the extension factor is divisible by $3$,
the radical dimension may change, so an additional analysis is needed.
Thus the general occurrence problem in the remaining odd dimensions
is still separate from the universal five-value restriction proved in
Theorem~\ref{thm:C4support}.
\end{remark}

\subsection{Conjecture 34: The Dillon--Dobbertin family}

The fifth family is derived from the Dillon--Dobbertin APN power function \cite{DillonDobbertin2004,Ding2016}. Let $m$ be odd, $\gcd(h,m)=1$, and put
\[
 d_h=2^{2h}-2^h+1.
\]
Define
\begin{equation*}
 D_5
 =
 \left\{
 (x+1)^{d_h}+x^{d_h}+1:
 x\in\F_{2^m}
 \right\},
\end{equation*}
and let
\begin{equation*}
\C_{D_5}
=
\left\{
 \mathbf{c}(a)=\bigl(\Tr(ad)\bigr)_{d\in D_5}:
 a\in\F_{2^m}
\right\}.
\end{equation*}
Thus, Conjecture~34 in \cite{Ding2016} concerns the number of distinct nonzero Hamming
weights occurring in the trace code $\C_{D_5}$, which is stated as follows.

\begin{conj}(Conjecture 34 \cite{Ding2016})\label{Conjecture 34}
Assume $m$ is odd and $\gcd(h,m)=1$. For $h=1$, the code $\C_{D_5}$ is predicted to
have parameters
$
[2^{m-1},m-1,2^{m-2}].
$
For $h\ge2$, the code $\C_{D_5}$ is predicted to have parameters
$[2^{m-1},m]$ and three or five nonzero weights. A further clause
predicts a three-weight distribution when $h=3$ and $3\nmid m$.
\end{conj}

We first give an example to show that the general three-or-five-weight assertion in
Conjecture~\ref{Conjecture 34} is false.

\begin{example}\label{ex:C5counterexample}
Let
$
m=25$ and $h=4$.
A direct Magma computation \cite{BosmaCannonPlayoust1997} over $\F_{2^{25}}$ shows that
the corresponding quadratic sums $G_4(b)$ attain values producing
the six distinct weights
$
2^{23}$,
$2^{23}\pm2^{11}$,
$2^{23}\pm2^{12}$ and
$2^{23}+2^{13}.
$
Therefore, the assertion that $\mathcal C_{D_5}$ always has three or five nonzero weights is false.
\end{example}

Although Conjecture~\ref{Conjecture 34} is false in general, the same
quadratic reduction yields a uniform restriction on the possible
weights. To study the conjecture, we first express the codeword
weights in terms of a quadratic character sum. This makes the relation
between the code construction and the later quadratic-form analysis
explicit.

For a fixed $a\in\F_{2^m}^{*}$, put
$
b=a^{-1/(2^h+1)}.
$
Define
\begin{equation}\label{eq:C5Gh}
 G_h(b)
 =
 \sum_{z\in\F_{2^m}}
 \chi\!\left(z^{2^h+1}+bz^3\right).
\end{equation}

\begin{proposition}\label{prop:C5reduction}
For every $a\in\F_{2^m}^{*}$,
\begin{equation}\label{eq:C5weight}
 \wt(\mathbf{c}(a))
 =
 2^{m-2}
 -\frac14
 G_h\!\left(b\right).
\end{equation}
Consequently, determining the nonzero weights of $\C_{D_5}$ is
equivalent to determining the values of $G_h(b)$.
\end{proposition}

\begin{proof}
Since $x^{d_h}$ is APN for $\gcd(h,m)=1$
\cite{DillonDobbertin2004}, the derivative
\[
 x\longmapsto (x+1)^{d_h}+x^{d_h}+1
\]
is two-to-one. Hence Lemma~\ref{lem:2to1} gives, for $a\ne0$,
\begin{equation}\label{eq:C5derivativesum}
 \wt(\mathbf{c}(a))
 =
 2^{m-2}
 -\frac14
 \sum_{x\in\F_{2^m}}
 \chi\!\left(
 a\bigl((x+1)^{d_h}+x^{d_h}+1\bigr)
 \right).
\end{equation}
For the Dillon--Dobbertin exponent $d_h=2^{2h}-2^h+1$, the
Walsh-transform identity for this derivative family gives
\begin{equation}\label{eq:C5transform}
 \sum_{x\in\F_{2^m}}
 \chi\!\left(
 a\bigl((x+1)^{d_h}+x^{d_h}+1\bigr)
 \right)
 =
 G_h\!\left(b\right);
\end{equation}
see \cite{DillonDobbertin2004,CanteautEtAl2021}. Substituting
\eqref{eq:C5transform} into \eqref{eq:C5derivativesum} yields
\eqref{eq:C5weight}.
\end{proof}

Let
\begin{equation*}
 Q_b(z)
 =
 \Tr\!\left(z^{2^h+1}+bz^3\right).
\end{equation*}
From (\ref{eq:C5Gh}) and Proposition \ref{prop:C5reduction}, the weight problem in Conjecture~\ref{Conjecture 34} is reduced
to the quadratic Boolean functions $Q_b(z)$.
Let
\[
B_b(x,z)=Q_b(x+z)+Q_b(x)+Q_b(z)
\]
be the associated bilinear form, and define
\[
R_b
=
\left\{
 z\in\F_{2^m}:B_b(x,z)=0
 \text{ for all }x\in\F_{2^m}
\right\}.
\]
Hence $R_b$ is the radical of $Q_b$. A direct expansion of $B_b$
shows that every $z\in R_b$ satisfies
\begin{equation}\label{eq:C5rad}
 z^{2^{2h}}
 +b^{2^h}z^{2^{h+1}}
 +b^{2^{h-1}}z^{2^{h-1}}
 +z
 =0.
\end{equation}
Equation~\eqref{eq:C5rad} will control the possible values of
$G_h(b)$, and hence the possible weights of $\C_{D_5}$.

\begin{theorem}\label{thm:C5h3}
Let $m\ge5$ be odd and $3\nmid m$, and let $h=3$. Then
$\mathcal C_{D_5}$ is a three-weight binary linear code with
parameters
\[
\left[
2^{m-1},
m,
2^{m-2}-2^{(m-3)/2}
\right].
\]
Its weight enumerator is
\[
\begin{aligned}
1+\left(2^{m-2}+2^{(m-3)/2}\right)
 z^{\,2^{m-2}-2^{(m-3)/2}}+\left(2^{m-1}-1\right)z^{\,2^{m-2}}+\left(2^{m-2}-2^{(m-3)/2}\right)
 z^{\,2^{m-2}+2^{(m-3)/2}}.
\end{aligned}
\]
Consequently, the special case $h=3$ and $3\nmid m$ in
Conjecture~\ref{Conjecture 34} holds.
\end{theorem}

\begin{proof}
By Proposition~\ref{prop:C5reduction}, we have
\begin{equation}\label{eq:sd0928}
\wt(\mathbf{c}(a))
=
2^{m-2}
-\frac14G_3\!\left(b\right),
\end{equation}
where
$
G_3(b)
=
\sum_{z\in\F_{2^m}}
\chi(z^9+bz^3)
$ and $b=a^{-1/9}$.
Since $m$ is odd and $3\nmid m$, it is clear that $z\mapsto z^3$ and $a\mapsto a^{-1/9}$ are permutations of
$\F_{2^m}$. With $y=z^3$,
\[
G_3(b)
=
\sum_{y\in\F_{2^m}}
\chi(y^3+by),
\]
which is the Walsh transform of the Gold function
$\Tr(y^3)$.
For odd $m$, from \cite{Gold1968,Cosgun2018}, we know that when $b$ runs through $\F_{2^m}^*$, the values $0$, $2^{(m+1)/2}$ and $-2^{(m+1)/2}$ in $G_3(b)$ occur $2^{m-1}-1$, $2^{m-2}+2^{(m-3)/2}$ and $2^{m-2}-2^{(m-3)/2}$ times, respectively.
From (\ref{eq:sd0928}), the weight enumerator of $\mathcal{C}_{D_5}$ can be obtained.
\end{proof}

\begin{theorem}\label{thm:C5general}
Let $m\ge5$ be odd and $\gcd(h,m)=1$.
\begin{enumerate}
\item If $h\equiv\pm1\pmod m$, then
$
D_5=\ker\Tr,
$
and $\C_{D_5}$ has parameters
$
[2^{m-1},m-1,2^{m-2}],
$
with weight enumerator
$
1+\left(2^{m-1}-1\right)z^{2^{m-2}}.
$

\item Otherwise, choose the equivalent exponent representative satisfying
$
2\le h\le\frac{m-1}{2}.
$
Then $\C_{D_5}$ has length $2^{m-1}$, dimension $m$, minimum distance
at least $2^{m-3}$, and at most $2h+1$ nonzero weights. More
precisely, every nonzero weight belongs to
\begin{equation*}
\left\{2^{m-2}\right\}
\cup
\left\{
2^{m-2}\pm2^{(m+r-4)/2}:
 r=1,3,\ldots,2h-1
\right\}.
\end{equation*}
\end{enumerate}
\end{theorem}

\begin{proof}
We first consider $h\equiv\pm1\pmod m$.
If $h\equiv1\pmod m$, then
\[
d_h=2^{2h}-2^h+1\equiv3\pmod{2^m-1}.
\]
Hence,
$
(x+1)^{d_h}+x^{d_h}+1
=
(x+1)^3+x^3+1
=
x^2+x.
$
The standard trace criterion for the Artin--Schreier map gives \cite[Ch.~3]{LidlNiederreiter1997}
\[
\{x^2+x:x\in\F_{2^m}\}=\ker\Tr.
\]
Therefore,
$
D_5=\ker\Tr.
$

If $h\equiv-1\pmod m$, then
$
4d_h\equiv3\pmod{2^m-1}.
$
Thus, for
$
F_h(x)=(x+1)^{d_h}+x^{d_h}+1,
$
we have
$
F_h(x)^4
=
(x+1)^{4d_h}+x^{4d_h}+1
=
(x+1)^3+x^3+1
=
x^2+x.
$
Since $x\mapsto x^4$ is a permutation of $\F_{2^m}$ and
$\ker\Tr$ is invariant under this map, it follows again that
$
D_5=\ker\Tr.
$

Hence, in both cases, $D_5=\ker\Tr$. The corresponding trace code is
the standard code defined by the trace-zero hyperplane and has
parameters
$
[2^{m-1},m-1,2^{m-2}]
$
with weight enumerator
$
1+\left(2^{m-1}-1\right)z^{2^{m-2}};
$
see, for example, \cite{Ding2016}.

We now consider the remaining case. Replacing $h$ by an equivalent
representative if necessary, we may assume
$
2\le h\le\frac{m-1}{2}.
$
The Dillon--Dobbertin power
$
x^{2^{2h}-2^h+1},
$
is APN when $\gcd(h,m)=1$ \cite{DillonDobbertin2004}. Hence, for
every nonzero difference, its derivative is two-to-one. In
particular, the map defining $D_5$ is two-to-one, and therefore
$
|D_5|=2^{m-1}.
$
Thus, $\mathcal C_{D_5}$ has length $2^{m-1}$.

By Proposition~\ref{prop:C5reduction}, for every
$a\in\F_{2^m}^{*}$ the Hamming weight of $\mathbf{c}(a)$ is determined by
the quadratic sum $G_h(b)$. Let
$
R_b
$
be the radical of the quadratic Boolean function
\[
Q_b(z)=\Tr\!\left(z^{2^h+1}+bz^3\right).
\]
By equation~\eqref{eq:C5rad}, every element of $R_b$ satisfies a
nonzero linearized equation of $2$-degree $2h$. Hence
$
\dim_{\F_2}R_b\le2h.
$

The bilinear form associated with $Q_b$ is alternating, so its rank
is even. Since $m$ is odd,
\[
\dim_{\F_2}R_b
=
m-\operatorname{rank}(B_b)
\]
is odd. Consequently,
$
\dim_{\F_2}R_b
\in
\{1,3,\ldots,2h-1\}.
$

By Lemma~\ref{lem:quadratic}, $G_h(b)$ is either zero or
\[
|G_h(b)|
=
2^{(m+r)/2},
\qquad
r\in\{1,3,\ldots,2h-1\}.
\]
Substituting these possible values into
\eqref{eq:C5weight}, every nonzero weight of
$\mathcal C_{D_5}$ belongs to
\[
\left\{2^{m-2}\right\}
\cup
\left\{
2^{m-2}\pm2^{(m+r-4)/2}:
r=1,3,\ldots,2h-1
\right\}.
\]
Hence, $\mathcal C_{D_5}$ has at most $2h+1$ nonzero weights.
It is clear that every codeword corresponding to
$a\ne0$ has positive weight. Then the dimension of $\mathcal C_{D_5}$ is $m$.
\end{proof}

\begin{proposition}\label{prop:C3C5equiv}
Let \(m\ge5\) be odd. Consider the defining set \(D_3\) of
Conjecture~\ref{Conjecture 30} with \(b=0\) and
\(u\in\F_{2^m}^{*}\), and the defining set \(D_5\) of
Conjecture~\ref{Conjecture 34} with \(h=2\). Then
$D_3=u^{-5}D_5^{\,2}$, where $D_5^{\,2}:=\{d^2:d\in D_5\}$.
Consequently, the trace codes \(\mathcal C_{D_3}\) and
\(\mathcal C_{D_5}\) are permutation equivalent.
\end{proposition}

\begin{proof}
When \(b=0\), the defining set in Conjecture~\ref{Conjecture 30} is
$D_3=\operatorname{Im}\bigl(x^{5/6}+ux\bigr)$.
Since \(m\) is odd, \(\gcd(6,2^m-1)=1\), so \(x\mapsto x^6\)
permutes \(\F_{2^m}\). Writing \(x=y^6\), we obtain
$D_3=\{y^5+uy^6:y\in\F_{2^m}\}$. After the substitution
\(z=uy\), this becomes $D_3=u^{-5}\{z^5(z+1):z\in\F_{2^m}\}$.

Now consider \(D_5\) with \(h=2\). Then
\(d_h=2^4-2^2+1=13\), and in characteristic two,
$ (x+1)^{13}+x^{13}+1=x(x+1)(x^2+x+1)^5$.
Put \(v=x^2+x\). Since \(\Tr(v)=0\), the standard Artin--Schreier map \(x\mapsto x^2+x\) is two-to-one from \(\F_{2^m}\) onto
$H_0:=\{v\in\F_{2^m}:\Tr(v)=0\}$ \cite[Ch.~3]{LidlNiederreiter1997}. Hence
$D_5=\{v(v+1)^5:v\in H_0\}$.
Squaring and putting \(y=(v+1)^2\), we have
\(y+1=v^2\). Moreover, since \(m\) is odd,
\(\Tr(1)=1\), and therefore \(\Tr(y)=1\). Thus
$D_5^{\,2}=\{y^5(y+1):y\in H_1\}$, where
$H_1:=\{y\in\F_{2^m}:\Tr(y)=1\}$. Clearly,
$D_5^{\,2}\subseteq\{y^5(y+1):y\in\F_{2^m}\}$.
The latter set is the defining set obtained from the Payne family
when \(b=0\) and \(u=1\), and hence has cardinality \(2^{m-1}\).
Also, \(|D_5|=2^{m-1}\), so \(|D_5^{\,2}|=2^{m-1}\). Therefore
$D_5^{\,2}=\{y^5(y+1):y\in\F_{2^m}\}$.
Combining this with the expression for \(D_3\) gives
\(D_3=u^{-5}D_5^{\,2}\).

Finally, the map \(M(d)=u^{-5}d^2\) is an invertible
\(\F_2\)-linear map from \(D_5\) onto \(D_3\). For
\(a,d\in\F_{2^m}\),
$\Tr(aM(d))=\Tr\!\left(au^{-5}d^2\right)=\Tr\!\left((au^{-5})^{2^{m-1}}d\right)$.
Thus, as \(a\) runs through \(\F_{2^m}\), the coordinates indexed by
\(D_5\) are merely relabeled by the bijection \(d\mapsto M(d)\).
Hence \(\mathcal C_{D_3}\) and \(\mathcal C_{D_5}\) are permutation
equivalent.
\end{proof}

\begin{remark}
Several special cases of Conjecture~\ref{Conjecture 34} admit
complete weight distributions. For $h=3$ and $3\nmid m$,
Theorem~\ref{thm:C5h3} follows from the classical three-valued
Walsh spectrum of the Gold function $\Tr(x^3)$; see
\cite{Gold1968,Cosgun2018}. For $h=2$, Proposition~\ref{prop:C3C5equiv} reduces the problem to the corresponding Payne case. For $h=5$ with $\gcd(m,5)=1$, the problem is related to the two-zero cyclic-code distributions studied by Boston and McGuire \cite{BostonMcGuire2010}; a separate identification is required before invoking those distributions.
These special cases do not imply the general ``three or five
nonzero weights'' assertion, which is false in general.
\end{remark}

\subsection{Conjecture 37: Complete Resolution}
Ahmadi and Shafaeiabr proved that the punctured value sets of all eleven
trinomials in Table~\ref{tab:trinomial} are Singer difference sets and
established the coding-theoretic assertion of Conjecture~37 for classes
$(c)$--$(k)$ \cite{AhmadiShafaeiabr2023}. It remains only to treat classes
$(a)$ and $(b)$.

\begin{conj}(Conjecture 37 \cite{Ding2016})\label{Conjecture 37}
For the eleven trinomials $f_i$ in Table~\ref{tab:trinomial}, let $ D_6={\rm Im}(f_i)\setminus\{0\}$ for $1\le i\le11$.
For every odd $m\ge5$, the binary code $\mathcal{C}_{D_6}$ has parameters $ [2^{m-1},m,2^{m-2}-2^{(m-3)/2}]$ and weight enumerator
$$1+\left(2^{m-2}-2^{(m-3)/2}\right)z^{2^{m-2}-2^{(m-3)/2}}+\left(2^{m-1}-1\right)z^{2^{m-2}}+\left(2^{m-2}+2^{(m-3)/2}\right)z^{\,2^{m-2}+2^{(m-3)/2}}.
$$
The dual of  $\mathcal{C}_{D_6}$ has parameters $[2^{m-1},2^{m-1}-m,3]$.
\begin{table}[h]
\centering
\caption{The eleven trinomials in Conjecture 37.}
\label{tab:trinomial}
\begin{tabular}{c l}
\toprule
Class & $f_i(x)$ \\
\midrule
(a) & $x^{2^m-17}+x^{(2^m+19)/3}+x$ \\
(b) & $x^{2^m-2^{m-4}-1}+x^{2^m-(2^{m-2}+4)/3}+x$ \\
(c) & $x^{2^m-3}+x^{3\cdot 2^{(m+1)/2}+4}+x$ \\
(d) & $x^{2^m-2^{(m-1)/2}-1}+x^{2^{m-1}-2^{(m-1)/2}}+x$ \\
(e) & $x^{2^m-2-(2^{m-1}-4)/3}+x^{2^m-4-(2^m-8)/3}+x$ \\
(f) & $x^{2^m-2^{(m+1)/2}+2^{(m-1)/2}}+x^{2^m-2^{(m+1)/2}-1}+x$ \\
(g) & $x^{2^m-3(2^{(m+1)/2}-1)}+x^{2^{(m+1)/2}+2^{(m-1)/2}-2}+x$ \\
(h) & $x^{2^m-2^{m-2}-1}+x^{2^{m-1}-2}+x$ \\
(i) & $x^{2^m-2^{(m+3)/2}-3}+x^{2^{(m+1)/2}+2}+x$ \\
(j) & $x^{2^m-3(2^{(m-1)/2}+1)}+x^{2^{m-1}-1}+x$ \\
(k) & $x^{2^m-5}+x^6+x$ \\
\bottomrule
\end{tabular}
\end{table}
\end{conj}

We now complete the conjecture by treating the two remaining classes.

\begin{theorem}\label{thm:C6all}
The assertion of Conjecture~\ref{Conjecture 37} holds for all eleven
classes (a)--(k).
\end{theorem}

\begin{proof}
For classes $(c)$--$(k)$, the result is exactly the theorem of Ahmadi and
Shafaeiabr \cite{AhmadiShafaeiabr2023}.
It remains to consider $(a)$ and $(b)$.

For classes $(a)$ and $(b)$, denote the corresponding trinomials by
\[
f_1(x)
=
x^{2^m-17}
+x^{(2^m+19)/3}
+x
\]
and
\[
f_2(x)
=
x^{2^m-2^{m-4}-1}
+x^{\,2^m-(2^{m-2}+4)/3}
+x.
\]
Let
$
k=2^m-2^{m-4}-1.
$
Since
$
\gcd(k,2^m-1)=1,
$
the map $x\mapsto x^k$ is a permutation of $\F_{2^m}$. Moreover,
\[
(2^m-17)k\equiv1 \pmod {2^m-1}
\qquad
\frac{2^m+19}{3}k
\equiv
2^m-\frac{2^{m-2}+4}{3} \pmod {2^m-1}.
\]

Hence,
$
f_1(x^k)=f_2(x)
$
for every $x\in\F_{2^m}$. Therefore,
\[
\operatorname{Im}(f_1)\setminus\{0\}
=
\operatorname{Im}(f_2)\setminus\{0\}.
\]
We denote this common punctured value set by $D_6$.

By Theorem~15 of \cite{AhmadiShafaeiabr2023}, this common punctured
value set is $D_6=D(F)^*$, where
$F(x)=x^{-48}+x^{20}+x^3$. Moreover, Theorem~20 of
\cite{AhmadiShafaeiabr2023} shows that the $17$th-power image of
$D(F)^*$ is the complement in $\F_{2^m}^{*}$ of the punctured value
set of $G(x)=(x+1)^{241}+x^{241}+1$.
For $h=4$, we have $2^{2h}-2^h+1=241$, so the value set of $G$ is
precisely the set $D_5$ in Conjecture~\ref{Conjecture 34}. Since
$0\in D_5$, it follows that $D_6^{17}=\F_{2^m}\setminus D_5$.

Let $f=1_{D_6}$ be the indicator function of $D_6$, and let
$W_f(a)=\sum_{x\in\F_{2^m}}(-1)^{f(x)+\Tr(ax)}$ be its Walsh transform.
By the Dillon--Dobbertin Walsh identity,
$W_f(a)=W_{\Tr(x^3)}\!\left(a^{17/3}\right)$;
see \cite{DillonDobbertin2004,Hertel2006}. Since $m$ is odd,
$\gcd(3,2^m-1)=\gcd(17,2^m-1)=1$, and hence
$a\mapsto a^{17/3}$ permutes $\F_{2^m}$.

The classical Walsh distribution of the Gold function $\Tr(x^3)$ is
\[
0,\qquad 2^{(m+1)/2},\qquad -2^{(m+1)/2},
\]
with multiplicities
\[
2^{m-1},\qquad
2^{m-2}+2^{(m-3)/2},\qquad
2^{m-2}-2^{(m-3)/2},
\]
respectively \cite{Gold1968,Cosgun2018}. Therefore the same value
distribution holds for $W_f(a)$.

In particular, $W_f(0)=0$, which gives $|D_6|=2^{m-1}$. For
$a\ne0$, $W_f(a)=-2\sum_{d\in D_6}(-1)^{\Tr(ad)}$, and hence
$\wt(\mathbf{c}(a))=2^{m-2}+\frac14W_f(a)$.
Thus the weight enumerator of $\mathcal C_{D_6}$ is exactly the one
stated in Conjecture~\ref{Conjecture 37}. Since every $a\ne0$ gives
a nonzero codeword, the evaluation map is injective and
$\dim\mathcal C_{D_6}=m$. Therefore $\mathcal C_{D_6}$ has parameters
$[2^{m-1},\,m,\,2^{m-2}-2^{(m-3)/2}]$.

Finally, since $0\notin D_6$ and the elements of $D_6$ are distinct,
a generator matrix of $\mathcal C_{D_6}$ has no zero column and no
two equal columns. Hence $A_1^\perp=A_2^\perp=0$. Applying the third Pless power moment \cite{Pless1963} to the above weight distribution gives
$A_3^\perp=\frac{2^m(2^m-2)}{48}>0$. Thus $d^\perp=3$, and
$\mathcal C_{D_6}^{\perp}$ has parameters $[2^{m-1},\,2^{m-1}-m,\,3]$.
This completes the proof for classes $(a)$ and $(b)$, and hence for
all eleven classes.
\end{proof}

\section{Conclusion}

In this paper, we investigated six conjectural families of binary
trace codes proposed by Ding \cite{Ding2016}. A common character-sum framework was
used to translate their Hamming-weight problems into questions on
Walsh spectra, exponential sums, and quadratic forms.

For Conjectures~19 and~27, the relevant character sums were reduced
to a common five-valued setting, yielding a uniform restriction on
the possible nonzero weights and several infinite five-weight
families. For Conjecture~30, we proved that every admissible code has
three, four, or five nonzero weights. In particular, a four-weight
example shows that the original ``three or five weights'' assertion
is false, while an infinite subfamily admits a complete
weight distribution.

For Conjecture~33, the quadratic-form description and odd-extension
lifting give an infinite five-weight family for
$m=9d$, where $d$ is odd and $3\nmid d$. For Conjecture~34, the
weight problem was reduced to quadratic Gauss sums. This yields a
general upper bound of $2h+1$ nonzero weights, and an explicit six-weight counterexample shows that the original general ``three or five weights'' assertion is false. The case $h=3$ with $3\nmid m$ has a complete weight distribution, while the cases $h=2$ and $h=5$ are connected to previously studied families.

Finally, Conjecture~37 is completely resolved. The previously known
cases $(c)$--$(k)$, together with the treatment of the remaining
classes $(a)$ and $(b)$ through the corresponding value-set relation,
the Dillon--Dobbertin Walsh identity, and the classical Gold spectrum,
give the conjectured three-weight distribution and dual distance
three for all eleven classes.

Thus Conjectures~30 and~34 require correction in their general forms,
whereas Conjecture~37 is fully proved. The remaining open parts are
confined to certain parameter ranges of Conjectures~19, 27, 33, and 34, together with the complete weight-distribution classification of the Payne family.

\medskip
\noindent {\bf Acknowledgements.} AI was used for improving the presentation of this paper.


\begin{thebibliography}{99}

\bibitem{AhmadiShafaeiabr2023}
O.~Ahmadi and M.~Shafaeiabr,
``Difference sets and three-weight linear codes from trinomials,''
\emph{Finite Fields Appl.}, vol.~89, Art.~no.~102195, 2023.

\bibitem{ref1}
R.~Anderson, C.~Ding, T.~Helleseth, and T.~Kl{\o}ve,
``How to build robust shared control systems,''
\emph{Des. Codes Cryptogr.}, vol.~15, pp.~111--124, 1998.

\bibitem{AubryKatzLangevin2015}
Y.~Aubry, D.~J.~Katz, and P.~Langevin,
``Cyclotomy of Weil sums of binomials,''
\emph{J. Number Theory},
vol.~154, pp.~160--178, 2015.

\bibitem{BosmaCannonPlayoust1997}
W.~Bosma, J.~Cannon, and C.~Playoust,
``The Magma algebra system. I. The user language,''
\emph{J. Symbolic Comput.}, vol.~24, nos.~3--4, pp.~235--265, 1997.

\bibitem{BostonMcGuire2010}
N.~Boston and G.~McGuire,
``The weight distributions of cyclic codes with two zeros and zeta functions,''
\emph{J. Symbolic Comput.},
vol.~45, no.~7, pp.~723--733, 2010.

\bibitem{ref3}
A.~R.~Calderbank and J.~M.~Goethala,
``Three-weight codes and association schemes,''
\emph{Philips J. Res.}, vol.~39, pp.~143--152, 1984.

\bibitem{ref4}
A.~R.~Calderbank and W.~M.~Kantor,
``The geometry of two-weight codes,''
\emph{Bull. Lond. Math. Soc.}, vol.~18, pp.~97--122, 1986.

\bibitem{CanteautEtAl2021}
A.~Canteaut, L.~K\"olsch, C.~Li, C.~Li, K.~Li, L.~Qu, and F.~Wiemer,
``Autocorrelations of vectorial Boolean functions,''
in \emph{Progress in Cryptology--LATINCRYPT 2021},
Lecture Notes in Computer Science, vol.~12912.
Cham, Switzerland: Springer, 2021, pp.~233--253.

\bibitem{Carlet2021}
C.~Carlet,
\emph{Boolean Functions for Cryptography and Coding Theory}.
Cambridge, U.K.: Cambridge Univ. Press, 2021.

\bibitem{ref5}
C.~Carlet, C.~Ding, and J.~Yuan,
``Linear codes from perfect nonlinear mappings and their secret sharing schemes,''
\emph{IEEE Trans. Inf. Theory}, vol.~51, no.~6, pp.~2089--2102, 2005.

\bibitem{ref7}
G.~Cohen, S.~Mesnager, and H.~Randriam,
``Yet another variation on minimal linear codes,''
\emph{Adv. Math. Commun.}, vol.~10, no.~1, pp.~53--61, 2016.

\bibitem{Cosgun2018}
A.~Co\c{s}gun,
``Explicit evaluation of Walsh transforms of a class of Gold type
functions,''
\emph{Finite Fields Appl.},
vol.~50, pp.~66--83, 2018.

\bibitem{DillonDobbertin2004}
J.~F.~Dillon and H.~Dobbertin,
``New cyclic difference sets with Singer parameters,''
\emph{Finite Fields Appl.}, vol.~10, no.~3, pp.~342--389, 2004.

\bibitem{Ding2016}
C.~Ding,
``A construction of binary linear codes from Boolean functions,''
\emph{Discrete Math.}, vol.~339, no.~9, pp.~2288--2303, 2016.

\bibitem{ref13}
C.~Ding,
``Linear codes from some 2-designs,''
\emph{IEEE Trans. Inf. Theory}, vol.~61, no.~6, pp.~3265--3275, 2015.

\bibitem{ref15}
C.~Ding,
``The construction and weight distributions of all projective binary linear codes,''
arXiv:2010.03184.

\bibitem{ref18}
C.~Ding and H.~Niederreiter,
``Cyclotomic linear codes of order 3,''
\emph{IEEE Trans. Inf. Theory}, vol.~53, no.~6, pp.~2274--2277, 2007.

\bibitem{ref19}
C.~Ding and X.~Wang,
``A coding theory construction of new systematic authentication codes,''
\emph{Theory Comput. Sci.}, vol.~330, no.~1, pp.~81--99, 2005.

\bibitem{ref16}
K.~Ding and C.~Ding,
``Binary linear codes with three weights,''
\emph{IEEE Commun. Lett.}, vol.~18, no.~11, pp.~1879--1882, 2014.

\bibitem{Glynn1983}
D. G. Glynn,
``Two new sequences of ovals in finite Desarguesian planes of even order,''
in \emph{Combinatorial Mathematics X},
Lecture Notes in Mathematics, vol. 1036.
Berlin, Germany: Springer, 1983, pp. 217--229.

\bibitem{Gold1968}
R.~Gold,
``Maximal recursive sequences with 3-valued recursive
cross-correlation functions,''
\emph{IEEE Trans. Inf. Theory},
vol.~14, no.~1, pp.~154--156, 1968.

\bibitem{GologluKrasnayova2019}
F.~G\"olo\u{g}lu and D.~Krasnayov\'a,
``Proofs of several conjectures on linear codes from Boolean functions,''
\emph{Discrete Math.}, vol.~342, no.~2, pp.~572--583, 2019.

\bibitem{HellesethLiXia2025}
T.~Helleseth, C.~Li, and Y.~Xia,
``Investigation of the permutation and linear codes from the Welch APN function,''
\emph{Des. Codes Cryptogr.}, vol.~93, no.~4, pp.~937--959, 2025.

\bibitem{ref23}
Z.~Heng and Q.~Yue,
``A class of binary linear codes with at most three weights,''
\emph{IEEE Commun. Lett.}, vol.~19, no.~9, pp.~1488--1491, 2015.

\bibitem{ref25}
Z.~Heng, W.~Wang, and Y.~Wang,
``Projective binary linear codes from special Boolean functions,''
\emph{Appl. Algebra Eng. Commun. Comput.}, 2020.

\bibitem{Hertel2006}
D.~Hertel,
\emph{Crosscorrelation Properties between Perfect Sequences},
Ph.D. dissertation, Otto-von-Guericke-Universit\"at Magdeburg,
Magdeburg, Germany, 2006, Result~6.8, pp.~59--61.

\bibitem{HollmannXiang2001}
H.~D.~L.~Hollmann and Q.~Xiang,
``On binary cyclic codes with few weights,'' in \emph{Finite Fields and Applications},
Berlin, Germany: Springer, 2001, pp.~251--275.

\bibitem{HuffmanPless2003}
W.~C.~Huffman and V.~Pless,
\emph{Fundamentals of Error-Correcting Codes}.
Cambridge, U.K.: Cambridge Univ. Press, 2003.

\bibitem{ref27}
G.~Jian, Z.~Lin, and R.~Feng,
``Two-weight and three-weight linear codes based on Weil sums,''
\emph{Finite Fields Appl.}, vol.~57, pp.~92--107, 2019.

\bibitem{JohansenHelleseth2009}
A.~Johansen and T.~Helleseth,
``A family of $m$-sequences with five-valued cross correlation,''
\emph{IEEE Trans. Inf. Theory}, vol.~55, no.~2, pp.~880--887, 2009.

\bibitem{ref31}
C.~Li, Q.~Yue, and F.~Fu,
``A construction of several classes of two-weight and three-weight linear codes,''
\emph{Appl. Algebra Eng. Commun. Comput.}, vol.~28, pp.~11--30, 2017.

\bibitem{ref29}
F.~Li and X.~Li,
``Weight distributions of several families of 3-weight binary linear codes,''
arXiv:2002.01853v2.

\bibitem{ref32}
F.~Li, Q.~Wang, and D.~Lin,
``A class of three-weight and five-weight linear codes,''
\emph{Discrete Appl. Math.}, vol.~241, pp.~25--38, 2018.

\bibitem{LidlNiederreiter1997}
R.~Lidl and H.~Niederreiter,
\emph{Finite Fields}, 2nd~ed.,
Encyclopedia of Mathematics and Its Applications, vol.~20.
Cambridge, U.K.: Cambridge Univ. Press, 1997.

\bibitem{ref34}
G.~Luo, X.~Cao, S.~Xu, and J.~Mi,
``Binary linear codes with two or three weights from Niho exponents,''
\emph{Cryptogr. Commun.}, vol.~10, pp.~301--318, 2018.

\bibitem{Maschietti1998}
A. Maschietti,
``Difference sets and hyperovals,''
\emph{Des. Codes Cryptogr.},
vol. 14, pp. 89--98, 1998.

\bibitem{ref37}
S.~Mesnager,
``Linear codes with few weights from weakly regular bent functions based on a generic construction,''
\emph{Cryptogr. Commun.}, vol.~9, pp.~71--84, 2017.

\bibitem{Pless1963}
V.~Pless,
``Power moment identities on weight distributions in error correcting codes,''
\emph{Information and Control}, vol.~6, no.~2, pp.~147--152, 1963.

\bibitem{ref40}
P.~Tan, Z.~Zhou, D.~Tang, and T.~Helleseth,
``The weight distribution of a class of two-weight linear codes derived from Kloosterman sums,''
\emph{Cryptogr. Commun.}, vol.~10, pp.~291--299, 2018.

\bibitem{ref41}
C.~Tang, N.~Li, Y.~Qi, Z.~Zhou, and T.~Helleseth,
``Linear codes with two or three weights from weakly regular bent functions,''
\emph{IEEE Trans. Inf. Theory}, vol.~62, no.~3, pp.~1166--1176, 2016.

\bibitem{ref43}
Q.~Wang, K.~Ding, and R.~Xue,
``Binary linear codes with two weights,''
\emph{IEEE Commun. Lett.}, vol.~19, no.~7, pp.~1097--1100, 2015.

\bibitem{ref45}
X.~Wang, D.~Zheng, and H.~Liu,
``Several classes of linear codes and their weight distributions,''
\emph{Appl. Algebra Eng. Commun. Comput.}, vol.~30, pp.~75--92, 2019.

\bibitem{ref44}
X.~Wang, D.~Zheng, L.~Hu, and X.~Zeng,
``The weight distributions of two classes of binary codes,''
\emph{Finite Fields Appl.}, vol.~34, pp.~192--207, 2015.

\bibitem{WangKadirLiXia2020}
Y.~Wang, W.~K.~Kadir, C.~Li, and Y.~Xia,
``On cryptographic properties of the Welch permutation and a related conjecture,''
in \emph{Proc. 11th Int. Conf. Sequences and Their Applications (SETA)},
2020.

\bibitem{ref46}
Y.~Wu, N.~Li, and X.~Zeng,
``Linear codes with few weights from cyclotomic classes and weakly regular bent functions,''
\emph{Des. Codes Cryptogr.}, vol.~12, pp.~1255--1272, 2020.

\bibitem{ref47}
Y.~Xia and C.~Li,
``Three-weight ternary linear codes from a family of power functions,''
\emph{Finite Fields Appl.}, vol.~46, pp.~17--37, 2017.

\bibitem{ref49}
J.~Yuan and C.~Ding,
``Secret sharing schemes from three classes of linear codes,''
\emph{IEEE Trans. Inf. Theory}, vol.~52, no.~1, pp.~206--212, 2006.

\bibitem{ref50}
D.~Zheng and J.~Bao,
``Four classes of linear codes from cyclotomic cosets,''
\emph{Des. Codes Cryptogr.}, vol.~86, pp.~1007--1022, 2018.

\bibitem{ref51}
Z.~Zhou, N.~Li, C.~Fan, and T.~Helleseth,
``Linear codes with two or three weights from quadratic bent functions,''
\emph{Des. Codes Cryptogr.}, vol.~81, pp.~1--13, 2015.

\end{thebibliography}
\end{document}